\documentclass[sigconf]{acmart}
\pdfoutput=1
\usepackage[title]{appendix}
\usepackage[linesnumbered,ruled,vlined]{algorithm2e}
\usepackage{subfigure}
\newtheorem{Def}{Definition}[section]

\renewcommand\footnotetextcopyrightpermission[1]{} 
 
\setcopyright{none}

\AtBeginDocument{%
  \providecommand\BibTeX{{%
    \normalfont B\kern-0.5em{\scshape i\kern-0.25em b}\kern-0.8em\TeX}}}

\allowdisplaybreaks

\begin{document}

\title{Spatio-Temporal Hierarchical Adaptive Dispatching for Ridesharing Systems}

\author{Chang Liu{$^\dagger$}, Jiahui Sun{$^\dagger$}, Haiming Jin{$^{\dagger}$}, Meng Ai{$^\S$}, Qun Li{$^\S$}, Cheng Zhang{$^\S$}, Kehua Sheng{$^\S$}, Guobin Wu{$^\S$}, Xiaohu Qie{$^\S$}, Xinbing Wang{$^\dagger$}}
\affiliation{$^\dagger$Shanghai Jiao Tong University, Shanghai, China \quad $^\S$Didi Chuxing, Beijing, China}

\renewcommand{\shortauthors}{Chang Liu, et al.}

\begin{abstract}
  Nowadays, ridesharing has become one of the most popular services offered by online ride-hailing platforms (e.g., Uber and Didi Chuxing). Existing ridesharing platforms adopt the strategy that dispatches orders over the entire city at a uniform time interval. However, the uneven spatio-temporal order distributions in real-world ridesharing systems indicate that such an approach is suboptimal in practice. Thus, in this paper, we exploit \textit{adaptive dispatching intervals} to boost the platform's profit under a guarantee of the maximum passenger waiting time. Specifically, we propose a hierarchical approach, which generates clusters of geographical areas suitable to share the same dispatching intervals, and then makes online decisions of selecting the appropriate time instances for order dispatch within each spatial cluster. Technically, we prove the impossibility of designing constant-competitive-ratio algorithms for the online adaptive interval problem, and propose online algorithms under partial or even zero future order knowledge that significantly improve the platform's profit over existing approaches. We conduct extensive experiments with a large-scale ridesharing order dataset, which contains all of the over 3.5 million ridesharing orders in Beijing, China, received by Didi Chuxing from October 1st to October 31st, 2018. The experimental results demonstrate that our proposed algorithms outperform existing approaches.
\end{abstract}

\maketitle
\section{Introduction}
Nowadays, the \textit{ridesharing} service offered by mobility-on-demand companies (e.g., Uber\footnote{https://www.uber.com/} and Didi Chuxing\footnote{https://www.didiglobal.com/}), which enables one vehicle to serve multiple orders simultaneously, has become an appealing alternative for traveling and commuting. On one hand, by sharing a ride with others, a passenger's mobility demand could usually be satisfied at a low price. On the other hand, ideally, ridesharing could dramatically decrease the overall number of vehicles on the road, which consequently helps reduce air pollution, conserve non-renewable energy resources, and alleviate traffic congestion. 

In practice, a ridesharing system is managed by a cloud-based platform that dispatches orders to vehicles after every fixed time interval (e.g., 2s) over the entire city \cite{Zheng2018PriceAware,Alonso2017PANS,Xu2018KDD}. However, such approach is usually suboptimal for the platform's profit, because in practice ridesharing orders are distributed unevenly both spatially and temporally, which is illustrated by the following Figures \ref{fig:all_city} and \ref{fig:heat}. As shown by Figure \ref{fig:all_city}, the number of ridesharing orders varies greatly over time. In the morning and evening peaks of weekdays during which most people commute, the number of orders is much more than that in the noon off-peak. On weekends, the peaks are much lower than those during weekdays, and there is even no obvious morning peak on Sunday. Furthermore, Figure \ref{fig:heat} shows that, even during the same periods, the number of ridesharing orders differs significantly across different geographical areas. 
\begin{figure}[h]
  \centering
  \includegraphics[width=\linewidth]{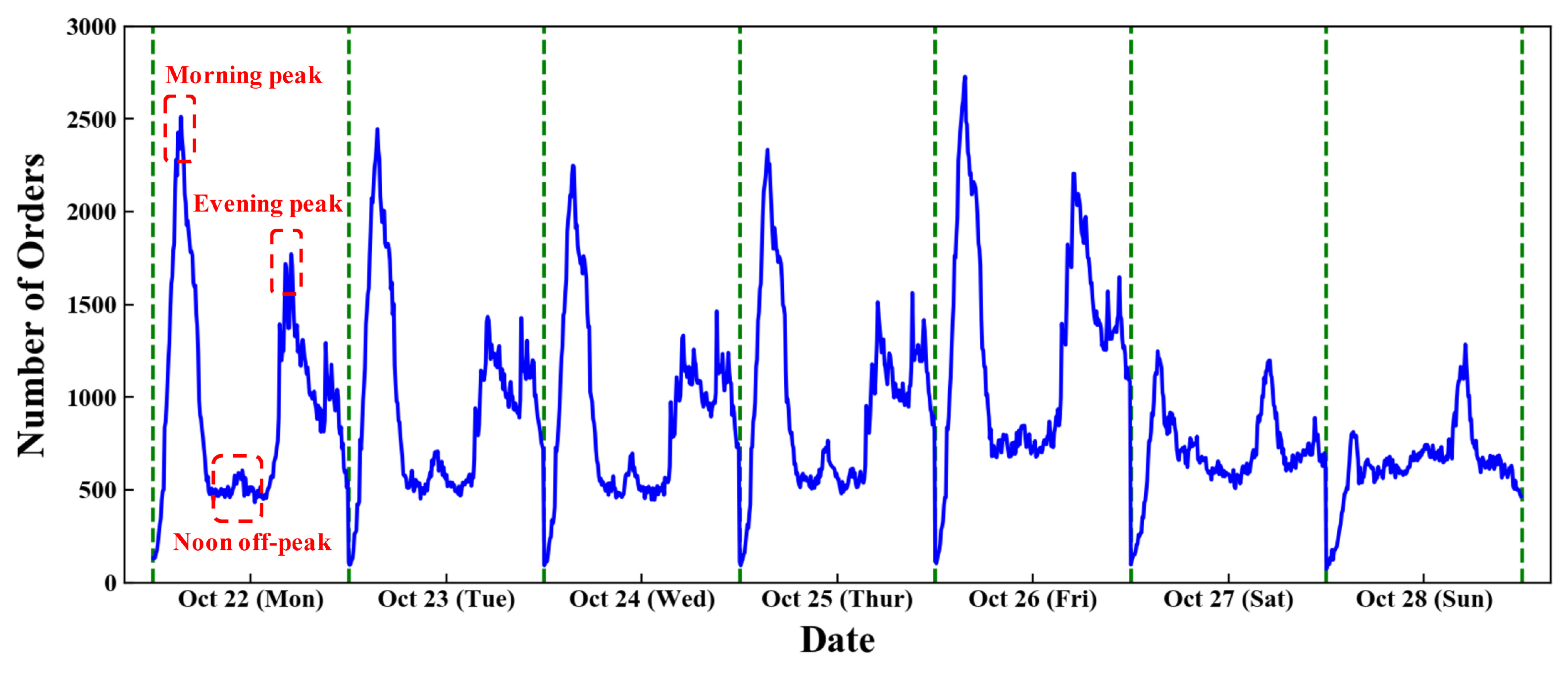}
  \vspace{-0.8cm}
  \caption{The number of ridesharing orders in Beijing received by Didi Chuxing in every 5 minutes from October 22nd to 28th, 2018.}
  \label{fig:all_city}
\end{figure}

\begin{figure}[h]
  \centering
  \subfigure[Morning peak]{
    \begin{minipage}[t]{0.3\linewidth}
      \centering
      \includegraphics[width=1.05\linewidth]{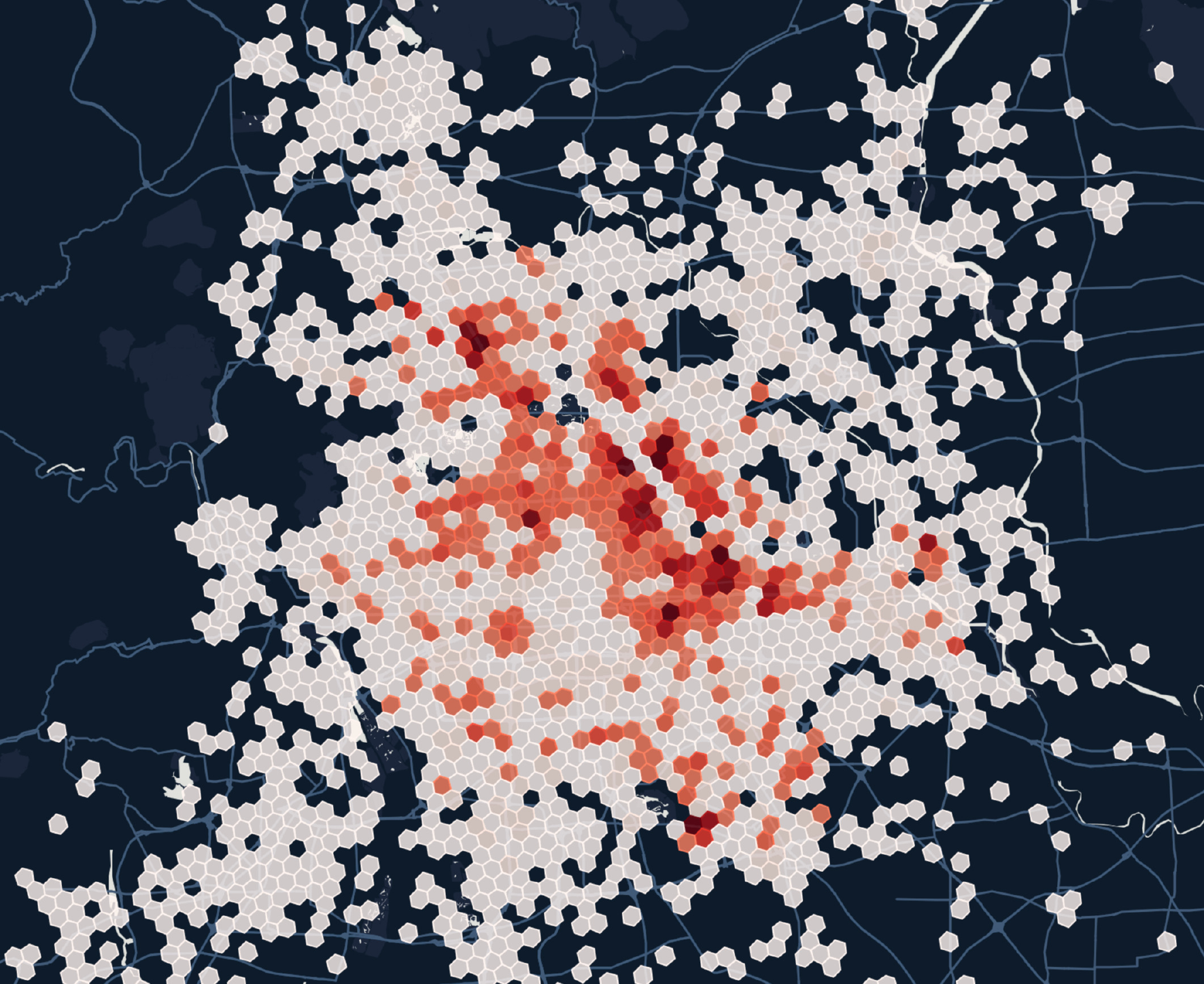}\vspace{-0.2cm}
      \label{fig:morn}
    \end{minipage}
  }
  \subfigure[Noon off-peak]{
    \begin{minipage}[t]{0.3\linewidth}
      \centering
      \includegraphics[width=1.05\linewidth]{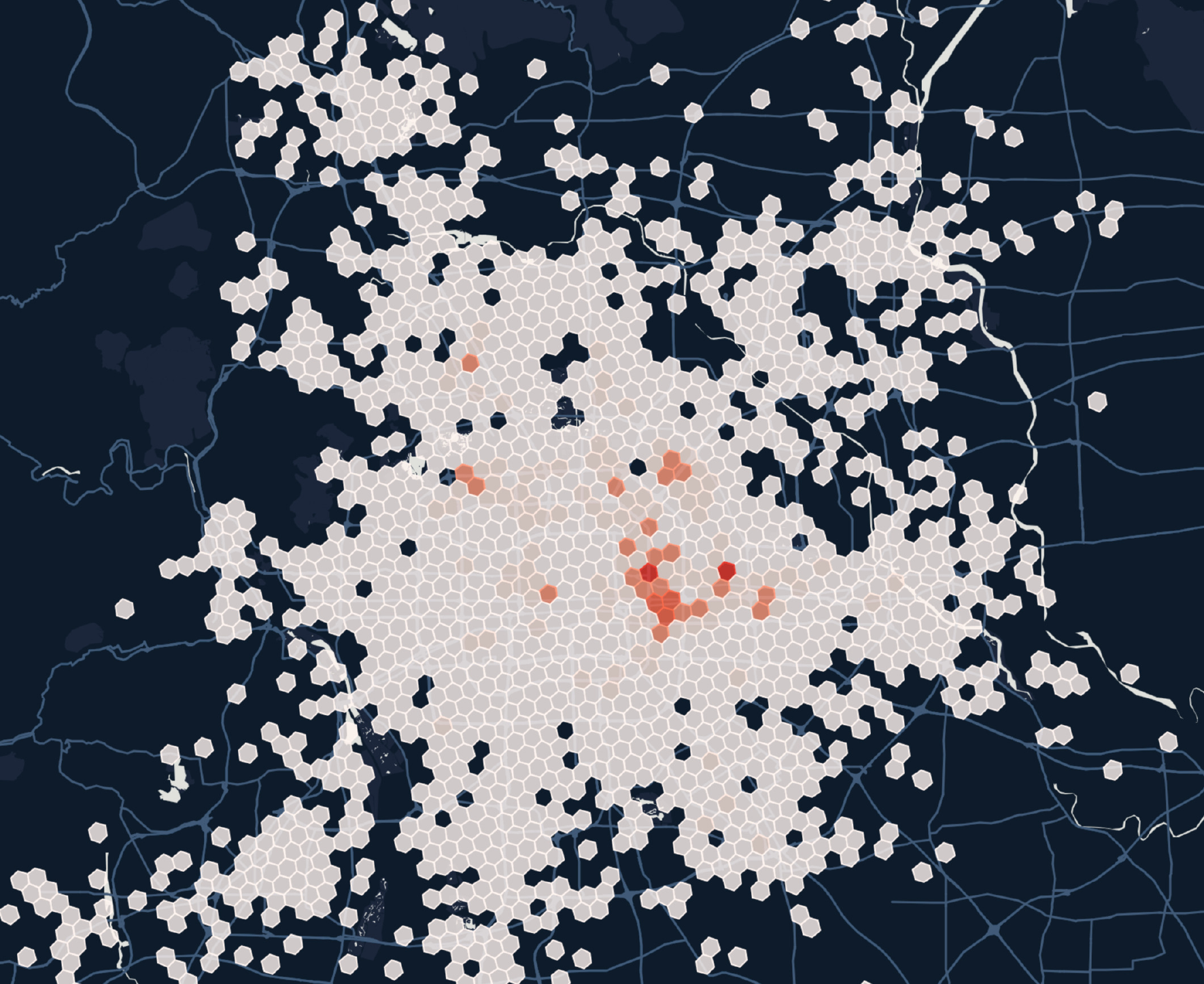}\vspace{-0.2cm}
      \label{fig:noon}
    \end{minipage}
  }
  \subfigure[Evening peak]{
    \begin{minipage}[t]{0.3\linewidth}
      \centering
      \includegraphics[width=1.05\linewidth]{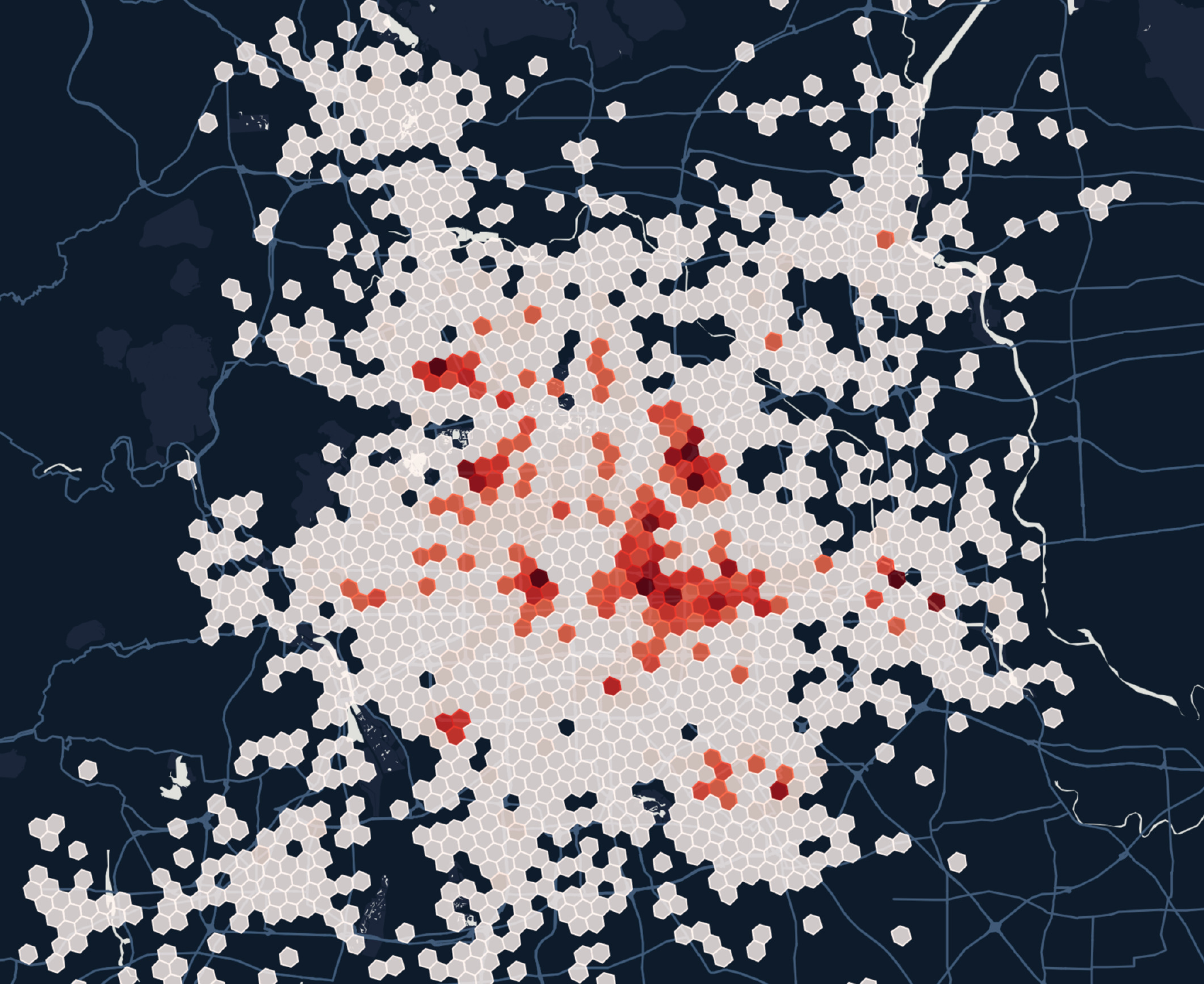}\vspace{-0.2cm}
      \label{fig:even}
    \end{minipage}
  }
  \centering
  \vspace{-0.4cm}
  \caption{The ridesharing order heat maps in Beijing during 3 different periods within one day, where the deeper the red color, the more the ridesharing orders.}
  \label{fig:heat}
\end{figure}

Naturally, at the time and area with denser orders, the platform could adopt a longer dispatching interval to fully exploit the possibility of encountering more shareable and profitable order combinations, which helps improve the platform's profit. However, an excessively long dispatching interval prolongs the passenger waiting time, which will inevitably incur more order cancellations that meanwhile deteriorate the platform's profit as well. Therefore, the dispatching intervals should appropriately adapt to the uneven spatio-temporal order distributions so as to strike a balance between receiving more shareable and profitable order combinations and the increase of order cancellations.

In this paper, dealing with the aforementioned trade-off, we propose a series of algorithms that \textit{generate adaptive dispatching intervals} which \textit{boost the platform's profit with a guaranteed maximum passenger waiting time}. However, designing such algorithms for real-world large-scale ridesharing systems is challenging in various aspects. Next, we would shed some light upon the philosophies behind how we address the various arising challenges. 

Towards achieving the above objectives, the first challenge is how to integrate both the spatial and temporal order distributions into decision making. In this paper, we propose a \textit{hierarchical} approach to address this challenge. Such an approach firstly divides areas suitable to share the same dispatching intervals into a cluster based on the spatial order distribution, and then dynamically decides the dispatching intervals in each spatial cluster by further taking into account the temporal order distribution. Furthermore, in practice, the arrivals and cancellations of ridesharing orders are highly stochastic and rather difficult to accurately predict. Such unpredictability inevitably requires decision making under uncertainty of future information, which is rather challenging. We address this challenge by casting our problem into the framework of \textit{optimal stopping}, which aims at deciding the optimal time instance to end the current dispatching interval in an online manner by leveraging sequentially fed historical order information. However, traditional solutions to optimal stopping cannot be directly applied in our problem setting due to the difficulty in reward calculation. We thus propose a suite of \textit{online adaptive interval algorithms}, which augment existing optimal stopping algorithms with our meticulously designed profit increment as the metric to calculate the reward. Using our adaptive interval algorithms, we are able to make online decisions of dispatching time instances that significantly improve the platform's profit with \textit{partial or even zero future order information}.

In summary, this paper makes the following contributions. 

\begin{itemize}
\item To the best of our knowledge, this paper is the first work that leverages the power of adapting dispatching intervals to the spatio-temporal order distribution for the objective of boosting the platform's profit under a maximum passenger waiting time for large-scale ridesharing.
\item We propose a novel hierarchical approach, which performs spatial clustering of geographical areas followed by online dispatching time instance decision. Specifically, we not only prove that it is impossible to solve the online adaptive interval problem with constant-competitive-ratio algorithms, but also propose a series of adaptive interval algorithms that make online decisions which significantly boost the platform's profit with only partial or zero future knowledge. 
\item We conduct extensive experiments based on a large-scale real-world ridesharing dataset, containing over 3.5 million ridesharing orders received by Didi Chuxing in Beijing, China, from October 1st to 31st, 2018. The experimental results validate the effectiveness of our proposed algorithms.
\end{itemize}

The rest of this paper is organized as follows. Section \ref{Preliminaries} first introduces our model and several fundamental definitions utilized in this paper, and then describes the problems we address and proves their hardness. In Section \ref{Spatial} and \ref{Temporal}, we describe and analyze our proposed algorithms. After describing our simulation results in Section \ref{Experiment}, and presenting the related work in Section \ref{Related}, we conclude this paper in Section \ref{Conclusion}.

\section{Preliminaries}\label{Preliminaries}
In this section, we introduce our ridesharing system model, formally describe the problems we solve, provide the corresponding mathematical formulations, and give an overview of the spatio-temporal hierarchical adaptive dispatching framework.

\subsection{Ridesharing System Model}
In this paper, we consider a ridesharing system managed by a cloud-based platform. In practice, after an order dispatch operation, the platform usually keeps receiving orders for a certain length of time before the next order dispatch. We then define the shortest duration between the platform's two consecutive dispatch operations as a \textit{unit time interval}, and denote it as $\Delta t$. In our model, the platform receives orders at any time, but only dispatches orders at the ends of unit time intervals. 

We define the \textit{active orders} at any time instance as the ones that are undispatched and have not yet been canceled by the passengers. Whenever the platform dispatches orders, it will call an order dispatch algorithm \cite{Alonso2017PANS,Bei2018AAAI,Jin2019CoRide,Li2019MeanField} which assigns active orders to available vehicles that are able to serve them. Clearly, in the order dispatch algorithm, it is essential to determine whether two orders are \textit{shareable} with each other, i.e., whether they can be served by the same vehicle simultaneously. In practice, the platform decides order shareability by jointly considering the detour distance, the pickup time, as well as various other factors. In this paper, we directly use DiDi Chuxing's order dispatch algorithm as a black box. Note that the type of the utilized order dispatch algorithm does not affect the design of the algorithms proposed in this paper. 

Given an active order set $\mathcal{A}$ as the input to the order dispatch algorithm, we use $R(\mathcal{A})$ to denote the \textit{platform's profit} for dispatching orders in $\mathcal{A}$, which is the difference between the gross income charged from all the orders in $\mathcal{A}$ that are served and the platform's payments to the drivers who serve them. 
\begin{figure}[h]
  \centering
  \subfigure[Area divided into cells.]{
    \begin{minipage}{0.48\linewidth}
      \centering
      \includegraphics[width=0.95\linewidth]{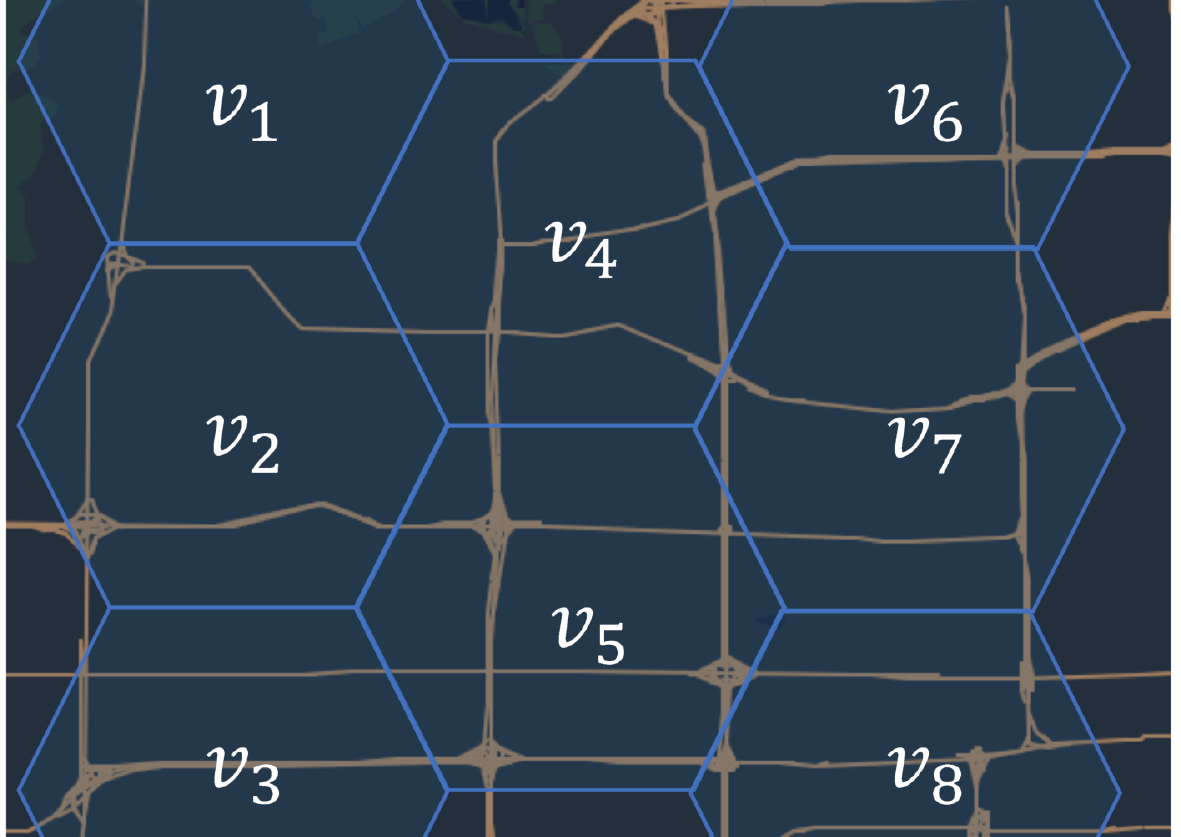}\vspace{0.1cm}
      \label{fig:graph}
    \end{minipage}
  }%
  \subfigure[Spatial shareability graph]{
    \begin{minipage}{0.48\linewidth}
      \centering
      \includegraphics[width=0.95\linewidth]{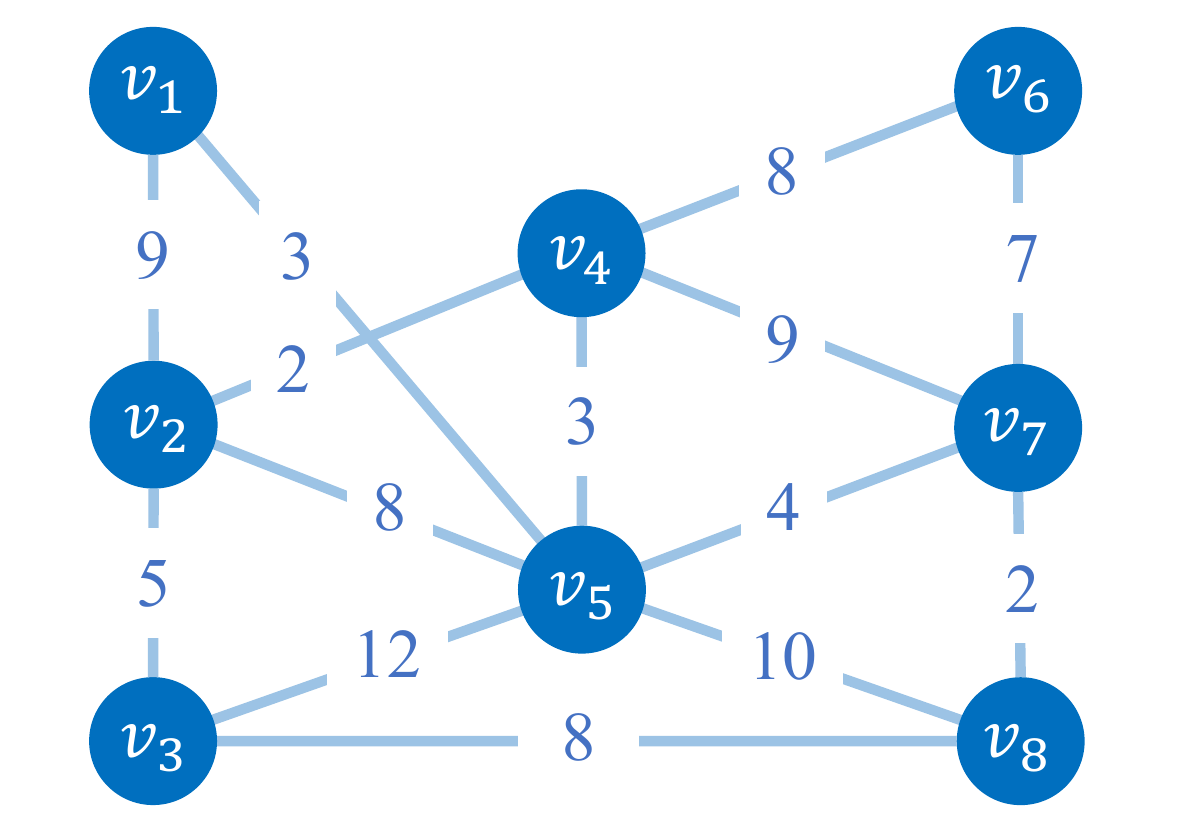}\vspace{0.1cm}
      \label{fig:graph2}
    \end{minipage}
  }%
  \centering
  \vspace{-0.4cm}
  \caption{Figure \ref{fig:graph} shows the geographical area that is divided into hexogon cells. Figure \ref{fig:graph2} shows the spatial shareability graph correponding to Figure \ref{fig:graph}, where the edge weights denote the number of shareable order pairs between two vertices. For example, there are 9 shareable order pairs between $v_1$ and $v_2$ in Figure \ref{fig:graph2}.}
\end{figure}

From the spatial perspective, we divide the entire city into equal-size hexagon cells\footnote{We ensure that the size of a cell satisfies that a driver is able to pick up any passenger in the cell where he locates in a duration short enough to be ignored, and meanwhile the spatial division is not overly fine-grained that introduces two many cells. Considering the above factors, we set the length of each cell as 0.8 km.}, and use a cell as the smallest unit that represents a location in our model. To capture the order shareability relationship among cells, we construct the \textit{spatial shareability graph} defined in the following Definition \ref{def:cellgraph}.

\begin{Def}[Spatial Shareability Graph]\label{def:cellgraph}
Given a historical order set $\mathcal{O}$, a spatial shareability graph is a weighted undirected graph $\mathcal{S}=(\mathcal{V}, \mathcal{E}, w)$, where each vertex $v\in \mathcal{V}$ represents a cell in the city, there exists an edge $(u,v) \in \mathcal{E}$, if there is at least one pair of shareable orders between cells $u$ and $v$ in the order set $\mathcal{O}$, and the function $w:\mathcal{E}\rightarrow\mathbb{Z}^+$ maps an edge $e=(u,v)$ to its weight $w(e)$ that represents the number of shareable order pairs between cells $u$ and $v$ in the order set $\mathcal{O}$. 
\end{Def}

By Definition \ref{def:cellgraph}, a spatial shareability graph is defined over a historical order set, and its edge weights correspond to the numbers of inter-cell shareable order pairs in the order set. Figures \ref{fig:graph} and \ref{fig:graph2} show an example of constructing a spatial shareability graph. Next, in the following Section \ref{sec:statement}, we present our formal statements and mathematical formulations of the problems we solve in this paper based on the models introduced in this section.

\subsection{Problem Statements and Formulations}\label{sec:statement}

Our ultimate objective is to \textit{maximize the platform's profit} over the planning horizon, by adapting the \textit{dispatching interval} (i.e., the duration between two consecutive order dispatches), according to both the spatial and temporal order distributions. 

We thus naturally decouple such problem into a \textit{spatial clustering (SC)} problem which clusters the spatial cells that are suitable to share the same dispatching intervals into the same group, and an \textit{adaptive interval (ADI)} problem which makes online decisions of whether to perform an order dispatch operation at the end of each unit time interval for each constructed spatial cluster. Then, in the rest of this section, we provide formal statements and formulations of the SC and ADI problems. 

\subsubsection{Spatial Clustering Problem}

In this paper, we refer to a connected subgraph of an spatial shareability graph $\mathcal{G}$ as a \textit{spatial cluster} of it. Next, we introduce the concept of \textit{spatial cluster set} in the following Definition \ref{def:ClusterSet}. 

\begin{Def}[Spatial Cluster Set]\label{def:ClusterSet}
Given a spatial shareability graph $\mathcal{G}=(\mathcal{V}, \mathcal{E}, w)$, a spatial cluster set $\mathcal{C}=\{\mathcal{G}_1, \mathcal{G}_2, \cdots\}$ with $\mathcal{G}_i=(\mathcal{V}_i, \mathcal{E}_i, w)$ is a set of spatial clusters (i.e., connected subgraphs) of $\mathcal{G}$, whose vertex sets form a partition of $\mathcal{V}$ (i.e., $\bigcup_{i:\mathcal{G}_i\in\mathcal{C}}\mathcal{V}_i=\mathcal{V}$, and $\mathcal{V}_i\bigcap\mathcal{V}_j=\emptyset, \forall i\not=j$). The set of all the spatial cluster sets of $\mathcal{G}$ is defined as $\mathcal{S}(\mathcal{G})$.
\end{Def}

By Definition \ref{def:ClusterSet}, a spatial cluster set divides the original spatial shareability graph into multiple non-overlapping subgraphs. Then, in the following Definition \ref{def:spatialcluster}, we formally describe the SC problem. 

\begin{Def}[SC Problem]\label{def:spatialcluster}
Given a spatial shareability graph $\mathcal{G}$, the SC problem aims to find the spatial cluster set $\mathcal{C}^*\in\mathcal{S}(\mathcal{G})$ with the maximum total edge weights such that the variance of the edge weights in each spatial cluster $\mathcal{G}_i^*\in\mathcal{C}^*$ is sufficiently small. 
\end{Def}

By such definition, we formulate the SC problem as the following optimization program \textbf{SC}. 
\vspace{-0.1cm}
\begin{align}
  \textbf{SC}:\max_{\mathcal{C}\in\mathcal{S}(\mathcal{G})}& \sum_{i:\mathcal{G}_i\in\mathcal{C}}\sum_{e\in\mathcal{E}_i}w(e)\label{eq:objective}\\
  \text{s.t. }
  &\sigma(\mathcal{E}_i) \leq \theta, \forall \mathcal{G}_i\in\mathcal{C},\label{upperbound}
\end{align}
where
\begin{equation*}\label{eq:costFunc}
\sigma(\mathcal{E}_i)=
\begin{cases}
\begin{aligned}
&\frac{1}{|\mathcal{E}_i|}\sum_{e\in\mathcal{E}_i}\big(w(e)\big)^2-\bigg(\frac{1}{|\mathcal{E}_i|}\sum_{e\in\mathcal{E}_i}w(e)\bigg)^2,&&\text{if~}\mathcal{E}_i\not=\emptyset\\
&0,&&\text{if~}\mathcal{E}_i=\emptyset
\end{aligned}
\end{cases}
\end{equation*}
is the variance of the edge weights in the edge set $\mathcal{E}_i$, $\theta$ denotes the upper bound that we set on such variance, and the objective function $\sum_{i:\mathcal{G}_i\in\mathcal{C}}\sum_{e\in\mathcal{E}_i}w(e)$ represents the sum of the edge weights in the constructed spatial cluster set $\mathcal{C}$. 

Our rationale of upper bounding the variance of the edge weights as in Constraint (\ref{upperbound}) is to ensure that each cluster consists of the cells with similar inter-cell shareability. Consequently, the cells in the same cluster is then suitable to share the same dispatching intervals. Furthermore, by maximizing the sum of the edge weights as in Objective Function (\ref{eq:objective}), we could avoid to the greatest extent that the cells with strong shareability with each other are divided into different clusters.

\subsubsection{Adaptive Interval Problem}After obtaining the spatial cluster set $\mathcal{C}^*$ by solving the optimization program \textbf{SC}, our next task is to generate adaptive dispatching intervals in each of the spatial cluster $\mathcal{G}_c^*\in\mathcal{C}^*$. We consider a planning horizon from time instance $t_0$ to $t_N$ which contains $N$ unit time intervals, and use $t_j$ to represent the ending time instance of the $j$th unit time interval. Next, we formally describe the ADI problem in the following Definition \ref{def:temprob}.

\begin{Def}[ADI Problem]\label{def:temprob}
For each spatial cluster $\mathcal{G}_i^*$, the ADI problem aims to make an online decision at each $t_j\in[t_0, t_N]$ on whether to dispatch orders with the knowledge of the current active order set, denoted as $\mathcal{A}_{i,j}$, and the previous dispatching time instances, with the objective to maximize the platform's profit over the entire planning horizon such that each  dispatching interval contains no more than $\beta$ unit time intervals.
\end{Def}

By such definition, we formulate the ADI problem for spatial cluster $\mathcal{G}_i^*$ as the following optimization program \textbf{ADI}$_i$. 
\vspace{-0.1cm}
\begin{align}
\textbf{ADI}_i:\max&~\sum_{j=1}^N x_j R(\mathcal{A}_{i,j})\\
\text{s.t.}&\sum_{j<k<l}x_j x_l(1-x_k)<\beta, && 1\leq j<l\leq N,\label{beta}\\
&x_j\in\{0,1\}, &&1\leq j \leq N,
\end{align}
where the variable $x_j=1$ means that the platform dispatches orders at time instance $t_j$ and otherwise $x_j=0$, Constraint (\ref{beta}) upper bounds the length of each dispatching interval by $\beta\Delta_t$, and the objective function $\sum_{j=1}^N x_j R(\mathcal{A}_{i,j})$ is exactly the platform's profit in cluster $\mathcal{G}_i^*$ over the planning horizon. 

Our rationale of introducing Constraint (\ref{beta}) is to upper bound the duration during which passengers wait before their orders are dispatched for the benefit of their experience. Furthermore, the above optimization program \textbf{ADI}$_i$ has to be solved in an online manner, because the active order set $\mathcal{A}_{i,j}$ for each time instance $t_j$ will not be revealed to the platform beforehand.

\subsection{Hierarchical Framework Overview}

To address the problems defined in Section \ref{sec:statement}, we propose the spatio-temporal hierarchical adaptive dispatching framework shown in Figure \ref{fig:overview}. 

\begin{figure}[t]
    \centering
    \includegraphics[width=\linewidth]{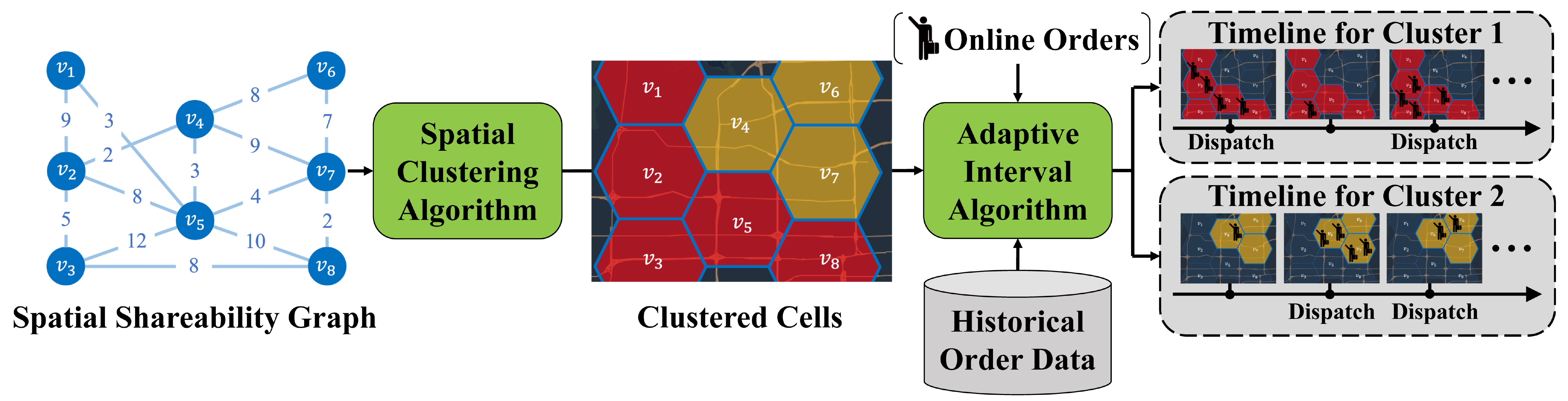}
    \caption{An overview of our spatio-temporal hierarchical adaptive dispatching framework.}
    \label{fig:overview}
    \vspace{-0.3cm}
\end{figure}

Such framework firstly runs the spatial clustering algorithm that solves the SC problem defined in Definition \ref{def:spatialcluster}, which takes as input the spatial shareability graph, and outputs a set of clusters such that the variance of the edge weights in each spatial cluster is sufficiently small. Then, the framework executes the adaptive interval algorithm that solves the ADI problem defined in Definition \ref{def:temprob}, which dynamically decides the dispatching intervals in each spatial cluster.

\section{Spatial Clustering Algorithm}\label{Spatial}
In this section, we describe and analyze our \textit{spatial clustering (SC) algorithm} presented in Algorithm \ref{algo:spatial} that solves the SC problem defined in Definition \ref{def:spatialcluster}, which is inspired by the clustering algorithm proposed in \cite{Oleksandr2006MST}.

\begin{algorithm}[h]
\KwIn{$\mathcal{G}=(\mathcal{V}, \mathcal{E}, w)$, $\theta$\;}
\KwOut{$\mathcal{C}$\;}
\tcp{\small Initialize an empty cluster set.}
$\mathcal{C} \leftarrow \emptyset$\;\label{algo:sp1}
\ForEach{connected component $\mathcal{S}_i$ of $\mathcal{G}$}{\label{algo:sp2}
    $T_i \leftarrow$ The maximum spanning tree of $\mathcal{S}_i$\;\label{algo:sp3}
    \tcp{\small Obtain the cluster set induced from $T_i$.}
    $\mathcal{C}_i \leftarrow$ Edge-Delete($T_i$, $\theta$)\;\label{algo:sp4}
    $\mathcal{C} \leftarrow \mathcal{C} \cup \mathcal{C}_i$\;\label{algo:sp5}
}
\KwRet $\mathcal{C}$\;\label{algo:sp6}
\caption{SC Algorithm}\label{algo:spatial}
\end{algorithm}

\begin{algorithm}
\KwIn{$T_i$, $\theta$\;}
\KwOut{$\mathcal{C}_i$\;}
$C_i\leftarrow \emptyset$\;\label{algo:sp7}
$E\leftarrow$ The edge set of $T_i$\;\label{algo:sp8}
\tcp{\small Get the variance of the edge weights.}
$\sigma(E) \leftarrow \frac{1}{|E|}\sum_{e\in E}\big(w(e)\big)^2 - \big(\frac{1}{|E|} \sum_{e\in E} w(e)\big)^2$\;\label{algo:sp9}
\uIf{$\sigma(E) \le \theta$}{\label{algo:sp10}
  $\mathcal{C}_i \leftarrow T_i$; \tcp{\small $T_i$ already forms a cluster.}\label{algo:sp11}
}
\Else{\label{algo:sp12}
  \tcp{\small Find the edge that decreases the variance the most if removed.}
  $e^* \leftarrow \arg\max_{e\in E}\sigma(E) - \sigma(E\backslash \{e\})$\;\label{algo:sp13}
  $\mathcal{T}_i \leftarrow$ The set of two trees after removing $e^*$ from $T_i$\;\label{algo:sp14}
  \ForEach{tree $T^\prime_j\in\mathcal{T}_i$}{\label{algo:sp15}
    \tcp{\small Obtain the cluster set induced from $T^\prime_j$.}
    $\mathcal{C}_j \leftarrow$ Edge-Delete($T^\prime_j$, $\theta)$\;\label{algo:sp16}
    $\mathcal{C}_i \leftarrow \mathcal{C}_i \cup \mathcal{C}_j$\;\label{algo:sp17}
  }
}
\KwRet $\mathcal{C}_i$\;\label{algo:sp18}
\caption{Edge-Delete Algorithm}\label{algo:ED}
\end{algorithm}

The inputs of Algorithm \ref{algo:spatial} include the spatial shareability graph $\mathcal{G} = (\mathcal{V},\mathcal{E}, w)$ and the threshold $\theta$ of the variance of the edge weights in each spatial cluster. Given the inputs, the \textit{SC} algorithm outputs a spatial cluster set $\mathcal{C}$. First of all, it initializes an empty cluster set $\mathcal{C}$ (line \ref{algo:sp1}). Then for each connected component $\mathcal{S}_i$ of $\mathcal{G}$, the algorithm finds its maximum spanning tree $T_i$ using the \textit{Prim’s algorithm} \cite{prim1957prim} (line \ref{algo:sp3}), and calls Algorithm \ref{algo:ED} to obtain the spatial cluster set $C_i$ induced from $T_i$ (line \ref{algo:sp4}-\ref{algo:sp5}), which is further added to $\mathcal{C}$ (line \ref{algo:sp5}). Finally, the algorithm returns the constructed spatial cluster set $\mathcal{C}$ (line \ref{algo:sp6}).

Algorithm \ref{algo:ED} which is a subroutine of Algorithm \ref{algo:spatial} takes as input a tree $T_i$ and the threshold $\theta$ which is originally fed to Algorithm \ref{algo:spatial}, and obtains the spatial clusters induced from $T_i$ recursively. It first initializes an empty cluster set $\mathcal{C}_i$ (line \ref{algo:sp7}), and gets the edge set $E$ of $T_i$ (line \ref{algo:sp8}). Next it calculates the variance of the edge weights in $E$ (line \ref{algo:sp9}). If the variance is no more than $\theta$, $T_i$ will be added as a spatial cluster (line \ref{algo:sp10}-\ref{algo:sp11}). Otherwise, the algorithm finds the edge that decreases the variance the most if removed, and removes that edge from $T_i$, which transforms the original tree $T_i$ into a set $\mathcal{T}_i$ of two trees (line \ref{algo:sp13}-\ref{algo:sp14}). For each tree $T^\prime_j$ in $\mathcal{T}_i$, the algorithm continues to obtain the cluster set $\mathcal{C}_j$ induced from $T^\prime_j$ recursively. This recursion repeats until the variances of the edge weights of all the trees induced from $T_i$ are no more than $\theta$ (line \ref{algo:sp16}). Then, the algorithm includes $\mathcal{C}_j$ into $\mathcal{C}_i$ (line \ref{algo:sp17}), and finally returns $\mathcal{C}_i$ as the constructed spatial cluster set of $T_i$ (line \ref{algo:sp18}).

The SC algorithm given in Algorithm \ref{algo:spatial} essentially  partitions the spatial shareability graph by deleting edges recursively until the variance of the edge weights in each cluster is sufficiently small. By such construction, the cells in each cluster have similar inter-cell shareablity, and are thus suitable to share the same dispatching intervals.

In practice, the order shareability among the spatial cells usually varies over time, which makes it necessary to construct different spatial shareability graphs in different periods of a day. The length of such period could be one or multiple hours depending on the historical order distributions. Note that the problem of deciding when to construct a new spatial shareability graph is not the focus of this paper, and does not affect the design of our proposed SC algorithm. Next, in the following Theorem \ref{complexityspatial}, we show that Algorithm \ref{algo:spatial} has a polynomial-time computational complexity.

\begin{theorem}\label{complexityspatial}
  Given a spatial shareability graph $\mathcal{G}=(\mathcal{V},\mathcal{E}, w)$ as an input, the worst case computational complexity of Algorithm \ref{algo:spatial} is $O\big(|\mathcal{V}|^3\big)$.
\end{theorem}

\begin{proof}
  Please refer to Appendix \ref{appendix:complexityspatial} for the detailed proof.
\end{proof}

\section{Adaptive Interval Algorithms}\label{Temporal}
In a real-world ridesharing system, the platform could assign two active orders to the same vehicle before their trips start, or assign an active order to share a vehicle with other orders that already start their trips. We refer to the former ridesharing mode as \textit{pre-trip} ridesharing, and the latter as \textit{in-trip} ridesharing. 

By such definition, we refer to the ADI problem with only pre-trip ridesharing as the \textit{pre-ADI} problem, and that with both pre-trip and in-trip ridesharing as the \textit{in-ADI} problem. In the rest of this section, we firstly analyze the hardness of the aforementioned pre- and  in-ADI problem, and then describe the details of the design and analysis of our proposed algorithms\footnote{Note that although the in-ADI problem focuses on a more general setting than the pre-ADI problem, we study the pre-ADI in this paper, because it serves as preliminaries for and sheds light upon how to address the in-ADI problem.}. Note that this section focuses on one single cluster extracted by the SC algorithm. Thus, we drop the subscript for cluster index for simplicity of presentation.

\subsection{Hardness Analysis}
In this section, we analyze the competitive hardness of the pre- and in-ADI problem. First, we prove the following Lemma \ref{lem1}, which serves as a preliminary result to prove Theorems \ref{thr1} and \ref{thr2}.

\begin{lemma}\label{lem1}
    There exists no deterministic algorithm that solves the pre-ADI problem with a constant competitive ratio.
\end{lemma}

\begin{proof}
Please refer to Appendix \ref{appendix:lem1} for the detailed proof. 
\end{proof}

Based on Lemma \ref{lem1}, we provide in the following Theorem \ref{thr1} the complete statement of the pre-ADI problem's competitive hardness. 

\begin{theorem}\label{thr1}
    There exists neither a deterministic, nor a randomized algorithm that solves the pre-ADI problem with a constant competitive ratio.
\end{theorem} 

\begin{proof}
Please refer to Appendix \ref{appendix:thr1} for the detailed proof. 
\end{proof}

Similarly, in the following Theorem \ref{thr2}, we show the competitive hardness of the in-ADI problem.

\begin{theorem}\label{thr2}
    There exists neither a deterministic, nor a randomized algorithm that solves the in-ADI problem with a constant competitive ratio.
\end{theorem}

\begin{proof}
Please refer to Appendix \ref{appendix:thr2} for the detailed proof. 
\end{proof}

Based on Theorems \ref{thr1} and \ref{thr2}, no algorithms could solve the pre- and in-ADI problem with a constant competitive ratio. We thus propose our online algorithms in the following Sections \ref{prealg} and \ref{sec:in-1/e}, which show good performances in our experiments.

\subsection{Algorithms for the pre-ADI Problem}\label{prealg}

Note that the pre-ADI problem belongs to the category of finite horizon optimal stopping problems, which aim to choose the optimal time instances to take specific actions based on sequentially obtained knowledge in a finite time horizon. Hence, we adapt the existing $1/e$ law, and backward induction-based algorithm for the optimal stopping problem into the \textit{$1/e$-pre-ADI algorithm} in Section \ref{sec:1/e} and \textit{BI-pre-ADI algorithm} in Section \ref{sec:BI}, respectively, to solve the pre-ADI problem.

\subsubsection{$1/e$-pre-ADI Algorithm}\label{sec:1/e}


Since the underlying temporal distribution of future orders is unknown to the platform, it is not feasible to directly optimize the platform's profit over the entire planning horizon at the end of each unit time interval. To address this problem, we propose to maximize the \textit{profit increment}, defined in the following Definition \ref{def:profitinc}, at each time instance instead. Before giving the formal definition of the profit increment, we define the platform's profit in any dispatching interval from $t_j$ to $t_k$ as $R_{j+1,k}$.

\begin{Def}[Profit Increment]\label{def:profitinc}
Given the last dispatching time instance being $t_l$, the increment in profit by postponing dispatching time by one unit time interval from $t_k$ to $t_{k+1}$, denoted by $P_{l,k,k+1}$, is defined as
\begin{equation*}
    P_{l,k,k+1} = R_{l+1,k+1}-(R_{l+1,k}+R_{k+1,k+1}).
\end{equation*}
Generally, we define the aggregate profit increment of postponing dispatching by one unit time interval at a time from $t_{l+1}$ to $t_j$ as
\begin{equation*}
    P_{l,l+1,j} = \sum_{k=l+1}^{j-1}P_{l,k,k+1} = R_{l+1,j} - \sum_{k=l+1}^j R_{k,k},
\end{equation*}
which is equivalent to the increment in profit between adopting the dispatching interval from $t_l$ to $t_j$ and dispatching uniformly at the end of each unit time interval between $t_l$ and $t_j$.
\end{Def}

By such definition, the profit increment $P_{l,k,k+1}$ can be negative, if the order cancellations in the $(k+1)$th unit time interval lead to more decrease in the profit. Thus, the total profit increment $P_{l,l+1,j}$ can decrease if postponing dispatching causes severer order cancellations, and it will increase if more shareable order pairs which can bring about more net profit are collected. The profit increment measures the increase in profit between dispatching by the current adaptive dispatching interval and dispatching by uniform dispatching intervals. By maximizing the profit increment, we increase the platform's profit to approach its maximum.

\begin{algorithm}
\KwIn{$\beta$, $\Delta t$, $t_N$, $t_l$, $t_c$, $P_{l,l+1,l+1},\cdots,P_{l,l+1,c}$\;}
\KwOut{$x_c$\;}
\tcp{\small Find the deadline for dispatching time.}
$t_{\max}\leftarrow \min\{t_N, t_l+\beta\Delta t\}$\;\label{algo:1/e1}
\tcp{\small Find the threshold time.}
$t_{\theta}\leftarrow t_l+\lceil\frac{t_{\max}-t_l}{e}\rceil$\;\label{algo:1/e2}
\uIf{$t_c < t_\theta$}{\label{algo:1/e3}
    $x_c\leftarrow 0$; \tcp{\small Threshold time not reached at $t_c$.} \label{algo:1/e4}
}
\uElseIf{$t_c = t_{\max}$}{\label{algo:1/e5}
    $x_c\leftarrow 1$; \tcp{\small Dispatching deadline reached at $t_c$.}\label{algo:1/e6}
}
\Else{\label{algo:1/e7}
    \uIf{$P_{l,l+1,c}>\max_{i\in \{l+1,\cdots, c-1\}}P_{l,l+1,i}$}{\label{algo:1/e8}
        $x_c\leftarrow 1$; \tcp{\small Profit increment maximized at $t_c$.} \label{algo:1/e9}
    }
    \Else{\label{algo:1/e10}
        $x_c\leftarrow 0$; \hfill \tcp{\small Profit increment not maximized at $t_c$.} \label{algo:1/e11}
    }
}
\KwRet $x_c$\;\label{algo:1/e12}
\caption{$1/e$-pre-ADI algorithm}\label{algo:1/e}
\end{algorithm}

Note that the pre-ADI problem is a variant of the classic secretary problem \cite{Gilbert:Mosteller:1966}. Motivated by the optimal stopping rule for classic secretary problem, we propose the \textit{$1/e$-pre-ADI algorithm} in Algorithm \ref{algo:1/e}, which chooses not to dispatch during the first $s$ unit time intervals, and afterwards dispatches at the end of the first unit time interval that has the maximum profit increment among all observed ones. According to the $1/e$ law \cite{Gilbert:Mosteller:1966}, we let $s=\left\lceil\frac{1}{e}\times\beta\right\rceil$, where $\beta$ represents the largest number of unit time intervals within a dispatching interval\footnote{Note that the parameter $\beta$ is introduced to ensure user experience by avoiding the dispatching interval to be excessively long.}, to increase the probability of dispatching at the optimal time instance.

The $1/e$-pre-ADI algorithm takes as inputs, the maximum number of unit time intervals in a dispatching interval $\beta$, the length of a unit time interval $\Delta t$, the beginning and ending time of the time period $t_0$ and $t_N$, the last dispatching time instance $t_l$, the current time instance $t_c$, and all profit increments in the current dispatching interval until the current time instance $P_{l,l+1,l+1},\cdots,P_{l,l+1,c}$. First of all, the algorithm finds the possible maximum ending time instance $t_{\max}$ of the current dispatching interval (line \ref{algo:1/e1}). Then, it obtains the threshold time instance $t_{\theta}$ (line \ref{algo:1/e2}). If the current time instance does not reach the threshold time instance, the algorithm decides not to dispatch regardless of the profit increment (line \ref{algo:1/e3}-\ref{algo:1/e4}). If the current time instance reaches the possible maximum ending time instance of the current dispatching interval, the algorithm decides to dispatch (line \ref{algo:1/e5}-\ref{algo:1/e6}). In other cases, Algorithm \ref{algo:1/e} compares the total profit increment of the current unit time interval $P_{l,l+1,c}$ with all the other profit increments in the current dispatching interval, and decides to dispatch if $P_{l,l+1,c}$ is the maximum, but not to dispatch otherwise (line \ref{algo:1/e7}-\ref{algo:1/e11}). 

In fact, Algorithm \ref{algo:1/e} decides whether to dispatch sequentially at the end of every unit time interval based on all orders raised since the last dispatch operation, and maximizes the probability of dispatching at the time instances that maximize the profit increment. Next, we analyze the computational complexity of Algorithm \ref{algo:1/e} in the following Theorem \ref{complexitye}. 

\begin{theorem}\label{complexitye}
    The computational complexity of Algorithm \ref{algo:1/e} is $O(1)$.
\end{theorem}
\begin{proof}
    Given all the profit increments calculated in advance, comparing it to all previous values in the current dispatching interval takes at most $\beta$ times, and all the other conditional statements can be executed in $O(1)$ time. Hence, the total time complexity of the $1/e$-pre-ADI algorithm is $O(1)$, which is constant.
\end{proof}

\subsubsection{BI-pre-ADI Algorithm}\label{sec:BI}
Different from the $1/e$-pre-ADI algorithm which is oblivious about the future, in this section, we propose the \textit{BI-pre-ADI algorithm} for the scenario where the platform has \textit{partial knowledge} on orders to be raised and canceled in the future. The partial knowledge that we incorporate refers to the prediction of the distribution of the future profits based on historical order data\footnote{Note that our BI-pre-ADI algorithm works with any prediction method, and exactly which method should be utilized to achieve the best prediction performance is out of the scope of this paper.}. 

\begin{algorithm}
\KwIn{$\beta$, $\Delta t$, $t_N$, $t_l$, $t_c$, $P_{l,l+1,c}$, $\overline{V}_{l, c+1}$\;}
\KwOut{$x_c$\;}
\tcp{Find the deadline for dispatching time.}
$t_{\max}\leftarrow \min\{t_l+\beta\Delta t, t_N\}$\;\label{algo:BI1}
\uIf{$t_c=t_{\max}$}{\label{algo:BI2}
    $x_c\leftarrow 1$; \tcp{\small Dispatching deadline reached at $t_c$.} \label{algo:BI3}
}
\Else{\label{algo:BI4}
    \uIf{$P_{l,l+1,c} \geq \overline{V}_{l, c+1}$}{\label{algo:BI5}
        $x_c\leftarrow 1$; \tcp{\small Higher profit increment than future.} \label{algo:BI6}
    }
    \Else{\label{algo:BI7}
        $x_c\leftarrow 0$; \tcp{\small Lower profit increment than future.} \label{algo:BI8}
    }
}
\KwRet $x_c$\;\label{algo:BI9}
\caption{BI-pre-ADI algorithm}\label{algo:BI}
\end{algorithm}


Given the current time instance $t_c$, and the last dispatching time instance $t_l$, we introduce a sequence of random variables, denoted by $V_{l,l+\beta}, \cdots, V_{l,c}$, which are defined in the following Definition \ref{def:value}.

\begin{Def}[Value Sequence]\label{def:value} 
We define the last term in the value sequence to be
\begin{equation}
    V_{l, l+\beta}=P_{l,l+1,l+\beta},
\end{equation}
which is the profit increment at $t_l+\beta\Delta t$, and then we inductively define the other terms in the sequence as
\begin{equation}
    V_{l, j}=\max\big\{P_{l,l+1,j}, \mathbb{E}[V_{l,j+1}]\big\},\forall j\in\{l+\beta-1,\cdots, c+1\}, 
\end{equation}
where, for each $j\in\{c+1,\cdots,l+\beta\}$, $P_{l,l+1,j}$ is a random variable at time instance $t_c$.
\end{Def}

By such definition, the value $V_{l,j}$ essentially represents the maximum profit increment the platform can possibly obtain starting from time instance $t_j$, and we use our estimate of $\mathbb{E}[V_{l, c+1}]$ from the historical order dataset\footnote{In this paper, we approximate $\mathbb{E}[V_{l, c+1}]$ by the average of historical values calculated from the historical order data, because ridesharing orders typically follow a strong temporal pattern. Note that other estimation methods could also be applied to obtain $\overline{V}_{l, c+1}$.}, denoted as $\overline{V}_{l, c+1}$, as one of the inputs to the BI-pre-ADI algorithm. 

Thus, apart from the same inputs taken by Algorithm \ref{algo:1/e}, the BI-pre-ADI algorithm given in Algorithm \ref{algo:1/e} also takes as inputs the profit increment of the current unit time interval $P_{l,l+1,c}$, and the estimate $\overline{V}_{l, c+1}$ of the expected value $\mathbb{E}[V_{l, c+1}]$. At the current time instance $t_c$, the algorithm firstly finds the possible maximum ending time instance of the current dispatching interval (line \ref{algo:BI1}). If the current time instance reaches the maximum ending time instance, the algorithm decides to dispatch orders immediately (line \ref{algo:BI2}-\ref{algo:BI3}). If not, it compares the observed profit increment of the current unit time interval $P_{l,l+1,c}$ with $\overline{V}_{l, c+1}$, and decides to dispatch only if $P_{l,l+1,c}$ is no less than $\overline{V}_{l, c+1}$ (line \ref{algo:BI4}-\ref{algo:BI8}). 

Next, in the following Theorem \ref{complexityBI}, we analyze the computational complexity of Algorithm \ref{algo:BI}.

\begin{theorem}\label{complexityBI}
    The computational complexity of Algorithm \ref{algo:BI} is $O(1)$.
\end{theorem}
\begin{proof}
    All the data preparation ought to be done in advance, and therefore contributes no time complexity. Looking up the stored value of $\overline{V}_{l, c+1}$ and making decisions whether to dispatch can be done in constant time. Hence, the total time complexity of the BI-pre-ADI algorithm is $O(1)$, which is constant.
\end{proof}

\subsection{Algorithm for the in-ADI Problem}\label{sec:in-1/e}
In this section, we consider the \textit{in-ADI problem}, which is even harder than the pre-ADI problem, because the orders which were dispatched in the previous dispatching intervals yet unfinished at the current time instance should still be considered in the current round of dispatch. 

As different divisions of past dispatching intervals could lead to different spatial distributions of unfinished dispatched orders, the profit and thus the optimal dispatching time of the current dispatching interval depend on past dispatching decisions. Similar to the pre-ADI problem, the in-ADI problem can also be categorized to a finite horizon optimal stopping problem. 

Next, we introduce the concept of \textit{in-ADI profit increment}, which will be utilized in our \textit{$1/e$-in-ADI algorithm} given in Algorithm \ref{algo:in-1/e} that solves the in-ADI problem. The reason for defining such concept is that the profit increment becomes dependent on past dispatching decisions when in-trip ridesharing is considered, which makes our original definition of the profit increment inaccurate. Note that we use $\mathbf{x}_{[1,j]}=(x_1,\cdots,x_j)$ to denote the vector containing the platform's dispatching decisions from $t_1$ to $t_j$, and $R_{l+1,j}(\mathbf{x}_{[1,j]})$ to denote the platform's profit in any dispatching interval from $t_l$ to $t_j$ given dispatching decisions $\mathbf{x}_{[1,j]}$.

\begin{Def} [in-ADI Profit Increment]
Given the last dispatching time instance $t_l$, and the current time instance $t_j$, the in-ADI profit increment in the current dispatching interval is defined as
\begin{equation}
    P_{l,l+1,j}(\mathbf{x}_{[1,j]}) = R_{l+1,j}(\mathbf{x}_{[1,j]})-R_{l+1,j}(\mathbf{x}^\prime_{[1,j]}),
\end{equation}
where $x_l=x_j=1$, $x_k = 0$ for $l<k<j$, $\mathbf{x}'_{[1,l-1]}=\mathbf{x}_{[1,l-1]}$, and $x^\prime_k=1$ for $l\leq k\leq j$. That is, $P_{l,l+1,j}(\mathbf{x}_{[1,j]})$ is the increment in the platform's profit when dispatching only at $t_j$ after $t_l$ compared with dispatching at the end of every unit time interval between $t_l$ and $t_j$.
\end{Def}
\vspace{-0.2cm}
\begin{algorithm}
\KwIn{$\beta$, $\Delta t$, $t_N$, $t_l$, $t_c$, $P_{l,l+1,l+1}(\mathbf{x}_{[1,l+1]})$,$\cdots$,$P_{l,l+1,c}(\mathbf{x}_{[1,c]})$\;}
\KwOut{$x_c$\;}
\tcp{\small Find the deadline for dispatching time.}
$t_{\max}\leftarrow \min\{t_N, t_l+\beta\Delta t\}$\;\label{algo:in-1/e1}
\tcp{\small Find the threshold time.}
$t_{\theta}\leftarrow t_l+\lceil\frac{t_{\max}-t_l}{e}\rceil$\;\label{algo:in-1/e2}
\uIf{$t_c < t_\theta$}{\label{algo:in-1/e3}
    $x_c\leftarrow 0$; \tcp{\small Threshold time not reached at $t_c$.} \label{algo:in-1/e4}
}
\uElseIf{$t_c = t_{\max}$}{\label{algo:in-1/e5}
    $x_c\leftarrow 1$; \tcp{\small Dispatching deadline reached at $t_c$.} \label{algo:in-1/e6}
}
\Else{\label{algo:in-1/e7}
    \uIf{$P_{l,l+1,c}(\mathbf{x}_{[1,c]})>\max_{j\in \{l+1,\cdots,c-1\}}P_{l,l+1,j}(\mathbf{x}_{[1,j]})$}{\label{algo:in-1/e8}
        $x_c\leftarrow 1$; \tcp{\small Profit increment maximized at $t_c$.} \label{algo:in-1/e9}
    }
    \Else{\label{algo:in-1/e10}
        $x_c\leftarrow 0$; \tcp{\small Profit increment not maximized at $t_c$.} \label{algo:in-1/e11}
    }
}
\KwRet $x_c$\;\label{algo:in-1/e12}
\caption{$1/e$-in-ADI algorithm}\label{algo:in-1/e}
\end{algorithm}

Algorithm \ref{algo:in-1/e} takes as its inputs the in-ADI profit increments $P_{l,l+1,l+1}(\mathbf{x}_{[1,l+1]})$,$\cdots$,$P_{l,l+1,c}(\mathbf{x}_{[1,c]})$, besides the same set of inputs taken by Algorithm \ref{algo:1/e}. The algorithm firstly finds the possible maximum ending time instance of the current dispatching interval, and the threshold time instance (line \ref{algo:in-1/e1}-\ref{algo:in-1/e2}). If the current time instance does not reach the threshold time instance, then it decides not to dispatch (line \ref{algo:in-1/e3}-\ref{algo:in-1/e4}). If the current time instance reaches the maximum ending time instance, then the algorithm decides to dispatch (line \ref{algo:in-1/e5}-\ref{algo:in-1/e6}). In other cases, Algorithm \ref{algo:in-1/e} compares the in-ADI profit increment of the current unit time interval with all in-ADI profit increments of previous unit time intervals in the current dispatching interval. Then, the algorithm decides to dispatch if the in-ADI profit increment of the current unit time interval is the maximum, and not to dispatch otherwise (line \ref{algo:in-1/e7}-\ref{algo:in-1/e11}).

Next, in the following Theorem \ref{complexityin-1/e}, we analyze the computational complexity of Algorithm \ref{algo:in-1/e}.

\begin{theorem}\label{complexityin-1/e}
    The computational complexity of Algorithm \ref{algo:in-1/e} is $O(1)$.
\end{theorem}

We omit the proof of this theorem, as it is similar to that of Theorem \ref{complexitye}.

\section{Experiments}\label{Experiment}
In this section, we describe the experiment setups, as well as the corresponding experimental results. 

\subsection{Experiment Setups}
\subsubsection{Datasets}\quad

\noindent
Our experiments are based on a real-world ridesharing order dataset, which contains all of the over 3.5 million ridesharing orders in the city of Beijing, China, received by Didi Chuxing from October 1st to October 31st, 2018. Each piece of data in the datasets has a number of features, such as \textit{order ID}, \textit{driver ID}, \textit{origin}, \textit{destination}, \textit{the time that the order was raised} and \textit{the cancellation time if the order was eventually canceled}. These features help us simulate and reproduce the behaviors of orders and drivers in the simulation.

\subsubsection{Simulator}\quad

\noindent
All orders in the simulation are reproduced according to our dataset. That is, orders' origins and destinations are set to be the nearest cells. The orders are raised in the simulation at the same timestamp as they were actually raised. 

Note that it is necessary to simulate order cancellations because one of the most important impacts of postponing dispatching is the massive increase in the quantity of canceled orders. 
We collect from the dataset all orders canceled by passengers. Each canceled order has a waiting time, which indicates how long the passenger waited before eventually canceling the order. We count the number of canceled orders in each unit time interval $[(i-1)\Delta t, i\Delta t)$, denoted by $\#\textnormal{canc}(i)$. Then, for each unit time interval $[(i-1)\Delta t, i\Delta t)$, we count the number of orders whose waiting time is longer than $(i-1)\Delta t$, denoted by $\#\textnormal{total}(i)$. Then, we estimate the cancellation probability as $\frac{\#\textnormal{canc}(i)}{\#\textnormal{total}(i)}$.
In our simulation, when deciding whether to dispatch at the end of each unit time interval, we calculate the waiting time of each order, find the corresponding cancellation probability, and cancel the order by the cancellation probability.

Drivers are retrieved by driver IDs and initialized at the same timestamps and locations when they first appear in our dataset. Timestamps and locations are mapped to the unit time intervals and vertices respectively. 
We update the positions of dispatched drivers at the end of each unit time interval by assuming that drivers deliver passengers along the planned routes\footnote{Note that we use the default route planning algorithms of Didi Chuxing in our experiment.} with an average speed obtained from the data. As for idle drivers, they walks randomly among the cells. 

\subsubsection{Baseline Algorithms}\quad

\noindent
In our experiments of the adaptive interval algorithms, we compare our algorithms with the \textit{uniform-base} algorithm that dispatches at a uniform time interval, which is widely adopted by existing ridesharing platforms. More specifically, we compare the $1/e$-pre-ADI and BI-pre-ADI algorithms with the uniform-base algorithm that dispatches orders at a uniform time interval $\Delta t$ considering only pre-trip ridesharing, whereas we compare the $1/e$-in-ADI algorithm with the same uniform-base algorithm that considers both pre-trip and in-trip ridesharing.


\subsection{Experimental Results}

\subsubsection{Experimental Results for the Spatial Clustering Problem}\quad

\noindent
We evaluate the spatial clustering algorithm on the entire urban area of Beijing, and construct the spatial shareability graphs in two typical periods of the day, $i.e.$, morning peak (7:00-10:00) and evening peak (17:00-20:00). 
In either period, the spatial shareability graph contains several connected components. Note that the connected components are connected in the spatial shareability graph, but are not necessarily neighbors geographically. Figures \ref{fig:morn_all}-\ref{fig:even_all} show the maximum spanning forests built from the spatial shareability graphs in the morning peak and the evening peak respectively. The generated maximum spanning forests are different in two periods due to the uneven temporal distribution of orders. The lines are thicker and darker for larger numbers of inter-cell shareable order pairs.

\begin{figure}[t]
  \centering
  \subfigure[Morning peak]{
    \begin{minipage}[t]{0.46\linewidth}
      \centering
      \includegraphics[width=\linewidth]{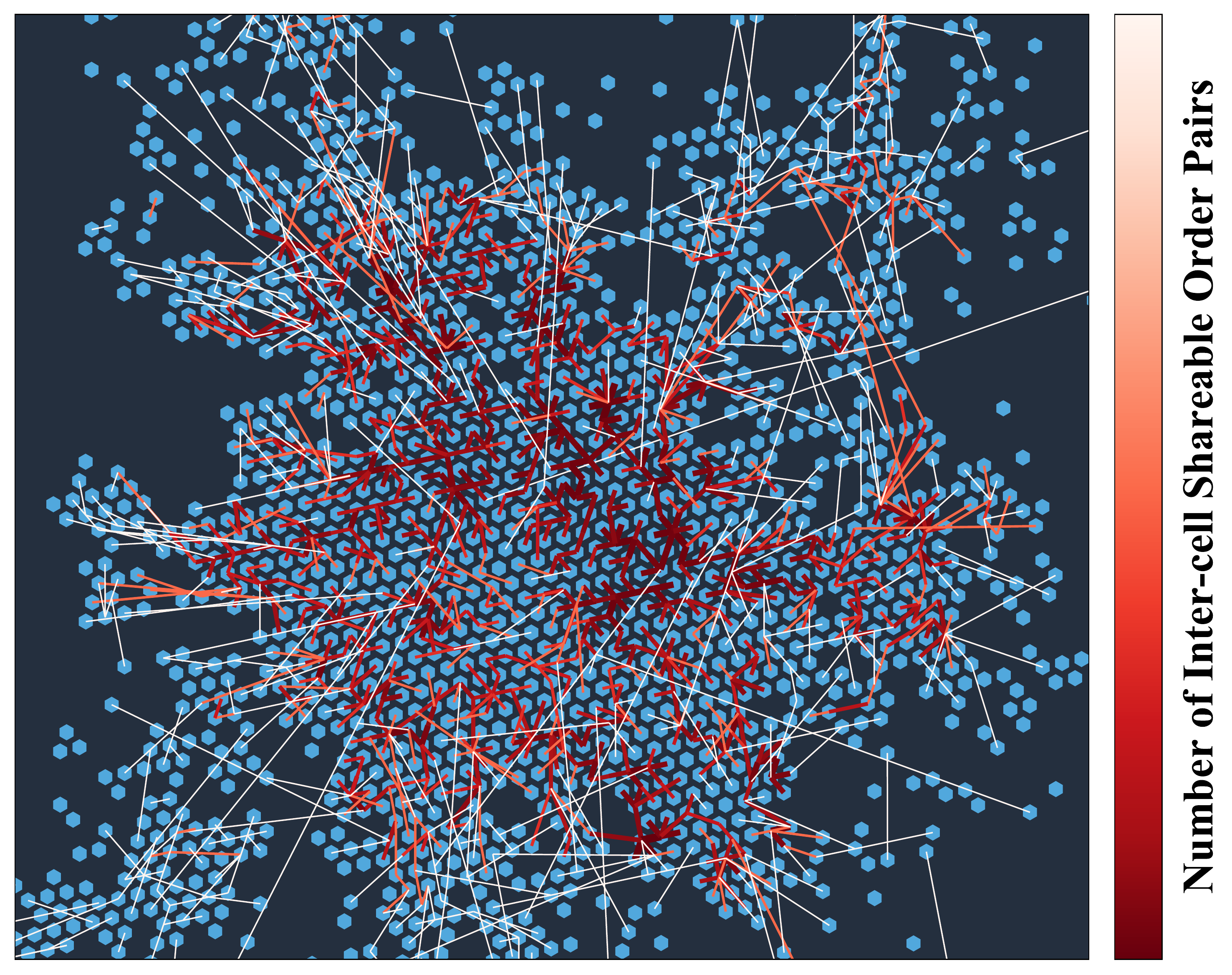}\label{fig:morn_all}
    \end{minipage}%
  }%
  \hspace{0.2cm}
  \subfigure[Evening peak]{
    \begin{minipage}[t]{0.46\linewidth}
      \centering
      \includegraphics[width=\linewidth]{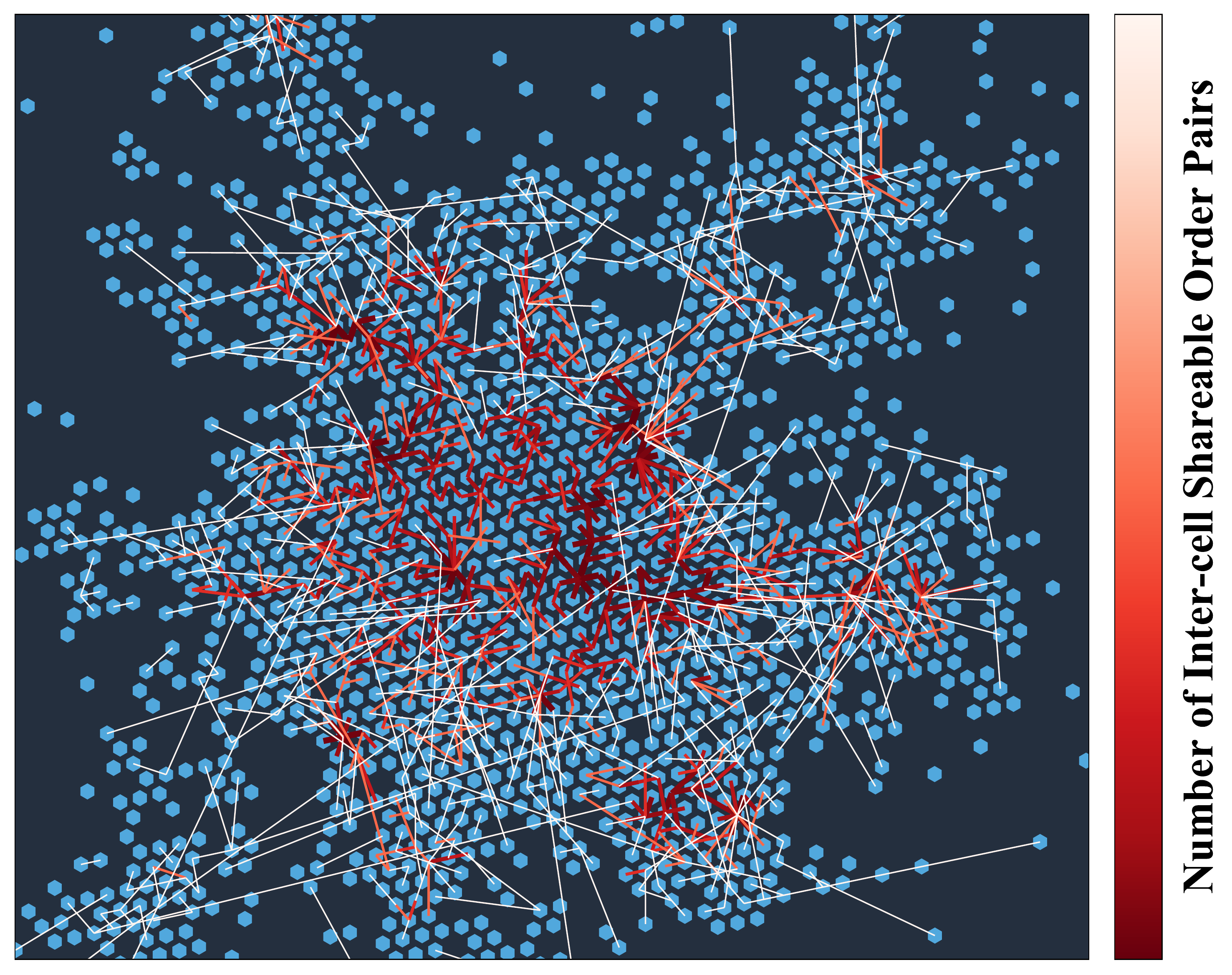}\label{fig:even_all}
    \end{minipage}
  }%
  \centering
  \vspace{-0.4cm}
  \caption{Maximum spanning forests of the morning peak and the evening peak.}
  \label{fig:three_periods}
  \vspace{-0.3cm}
\end{figure}

\begin{figure}[h]
  \centering
  \subfigure[Morning peak]{
    \begin{minipage}[t]{0.48\linewidth}
      \centering
      \includegraphics[width=\linewidth]{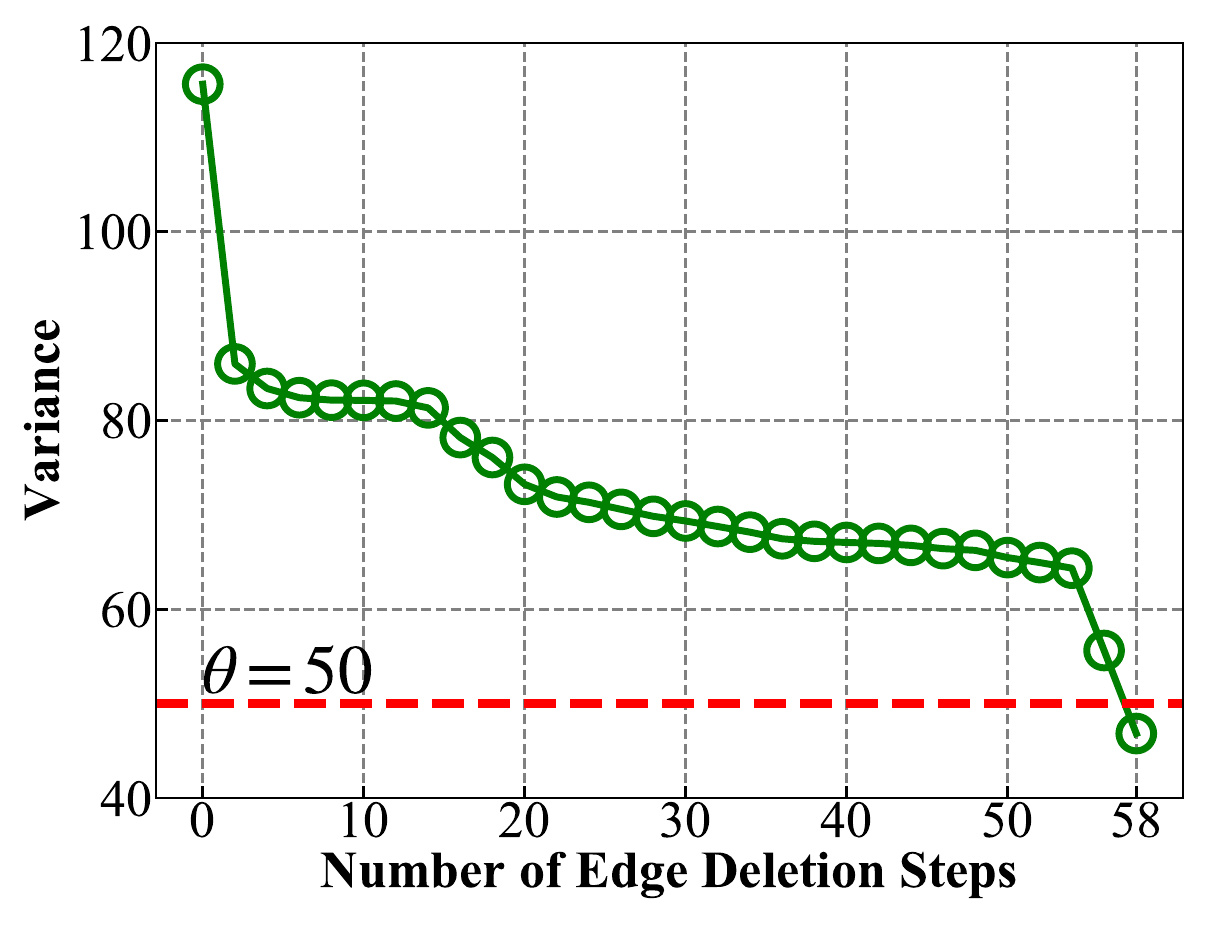}\label{fig:var_morn_all}
    \end{minipage}%
  }%
  \subfigure[Evening peak]{
    \begin{minipage}[t]{0.48\linewidth}
      \centering
      \includegraphics[width=\linewidth]{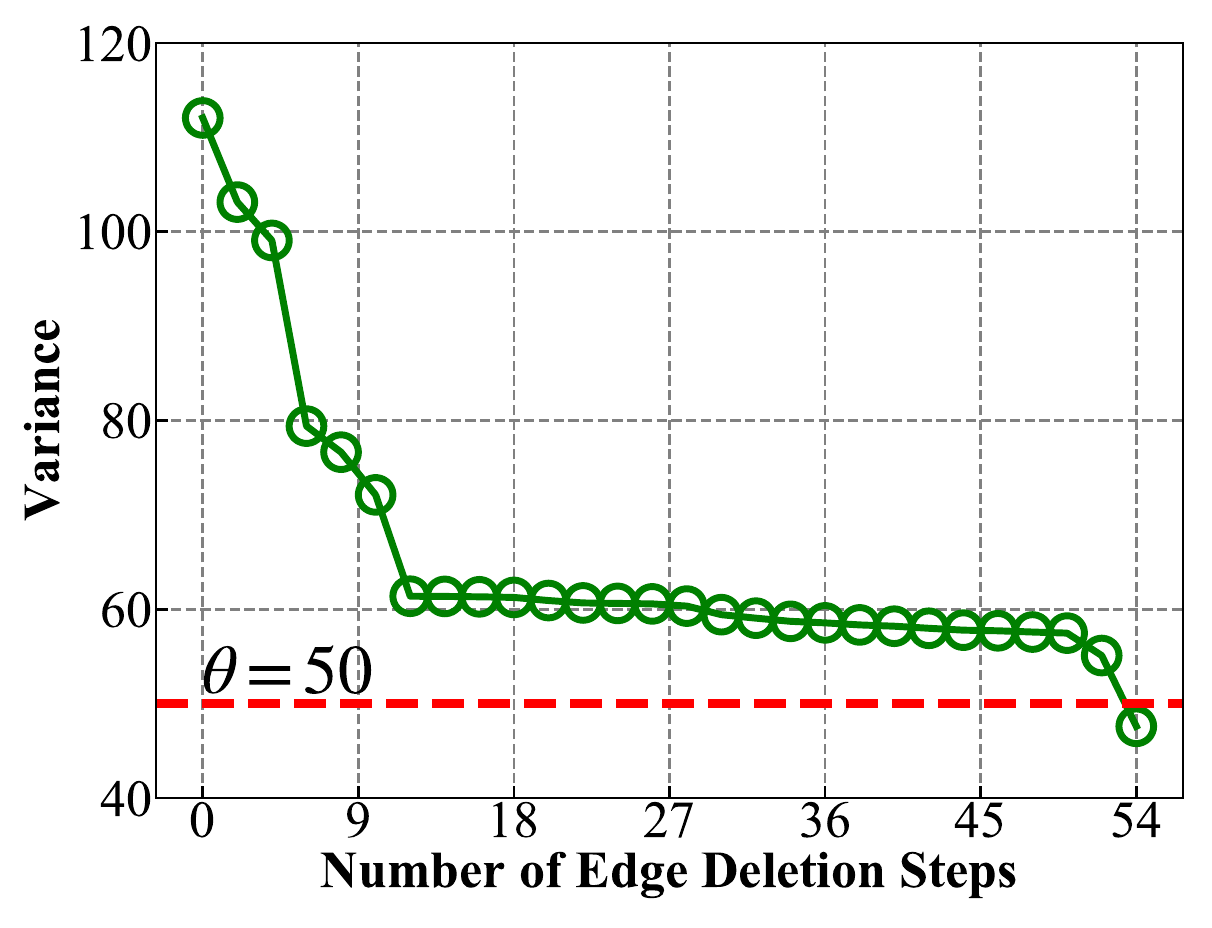}\label{fig:var_even_all}
    \end{minipage}%
  }%
  \vspace{-0.4cm}
  \caption{The maximum variance of the clusters in each edge deletion step.}
  \label{fig:var_morn_even_all}
\end{figure}

In either period, for each maximum spanning tree, the SC algorithm recursively deletes edges until the variance of the edge weights of each newly formed cluster becomes no greater than a manually selected threshold $\theta$. 
To ensure that each cluster consists of the cells with similar inter-cell shareability, $\theta$ is set as a small value 50 in our experiments. We record the maximum variance of the clusters in each edge deletion step, and show the decreasing trend of such maximum variance in Figures \ref{fig:var_morn_all} and \ref{fig:var_even_all} in the morning peak and the evening peak respectively. In Figures \ref{fig:var_morn_all} and \ref{fig:var_even_all}, it can be observed that the variance rapidly decreases down to the threshold, which shows the computational complexity of our spatial clustering algorithm is rather low.

Out of all clusters generated, we extract four representative clusters for further analysis and experiments without loss of generality. The clusters 1 and 2, which appear in the morning peak, are mainly composed of residential districts, where massive amounts of orders are submitted by passengers who are carpooling to work. Therefore, there are a number of inter-cell shareable order pairs within the clusters. Similarly, the clusters 3 and 4, which are formed in the commercial areas in the evening peak, have a large number of inter-cell shareable order pairs as well.

The above experimental results shown in this section demonstrate that our spatial clustering algorithm could successfully identify the spatial clusters as desired in a computationally efficient manner.

\subsubsection{Experimental Results for the Adaptive Interval Problem}\quad
\begin{figure}[t]
  \centering
  \subfigure[Cluster 1]{
    \begin{minipage}[t]{0.47\linewidth}
      \centering
      \includegraphics[width=\linewidth]{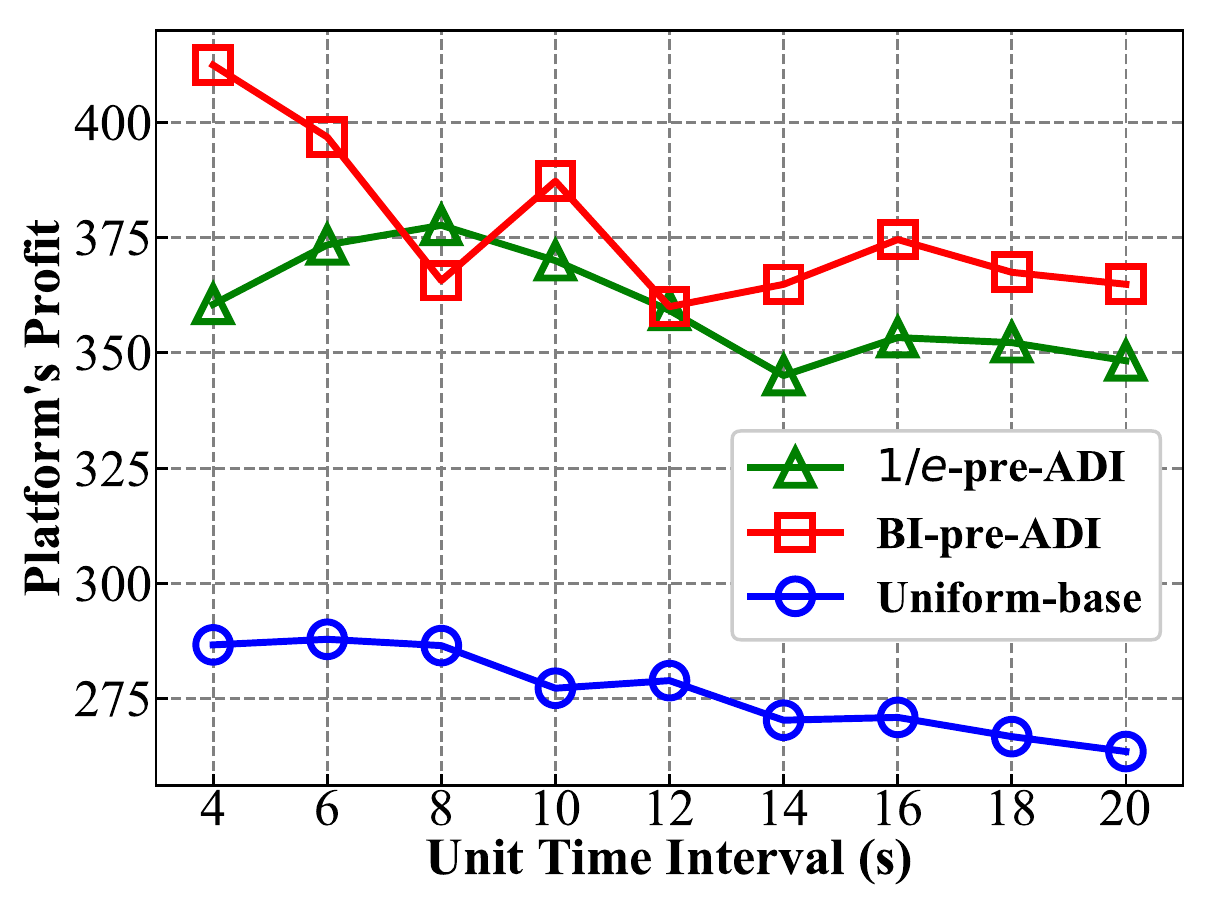}\label{int_revenue_hlg}
    \end{minipage}
  }
  \subfigure[Cluster 2]{
    \begin{minipage}[t]{0.47\linewidth}
      \centering
      \includegraphics[width=\linewidth]{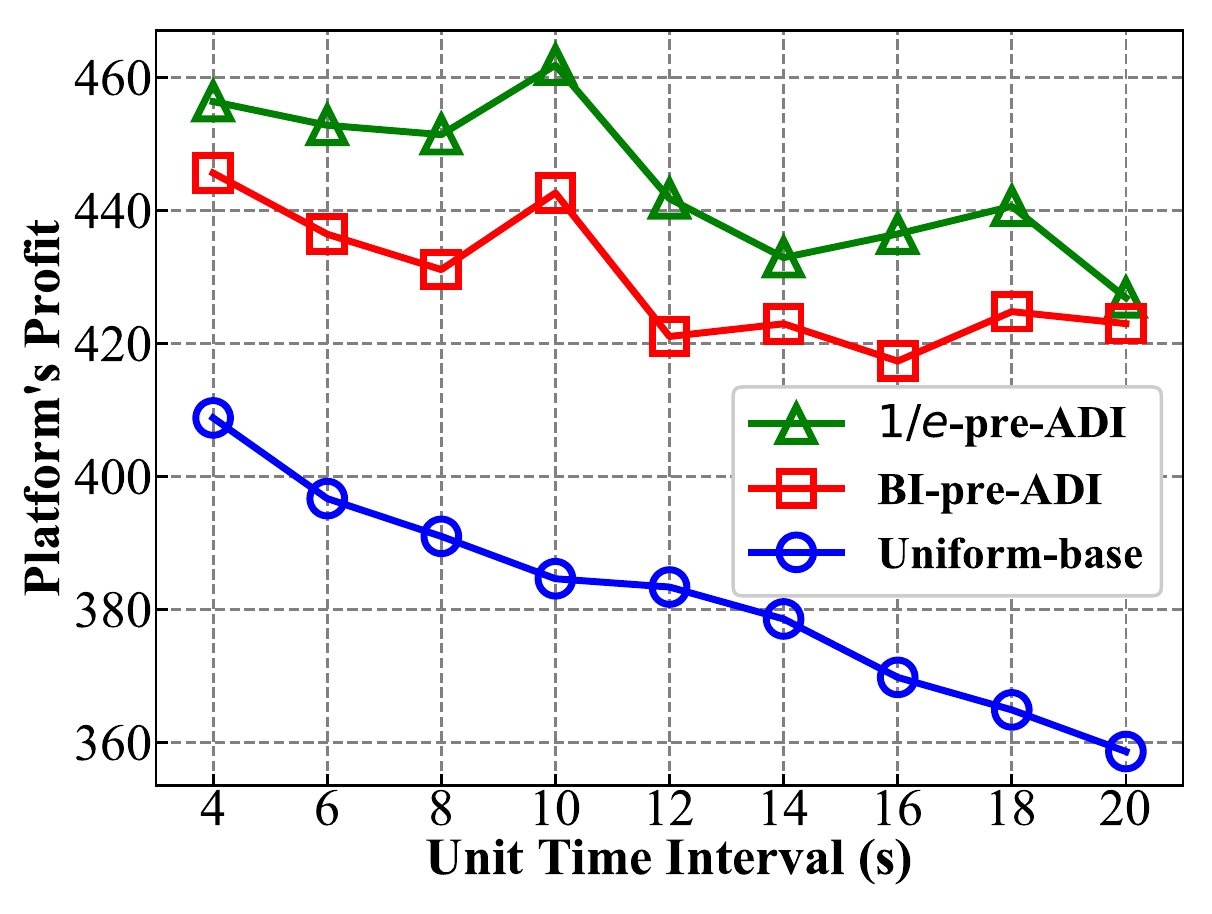}
    \end{minipage}
  }

  \subfigure[Cluster 3]{
    \begin{minipage}[t]{0.47\linewidth}
      \centering
      \includegraphics[width=\linewidth]{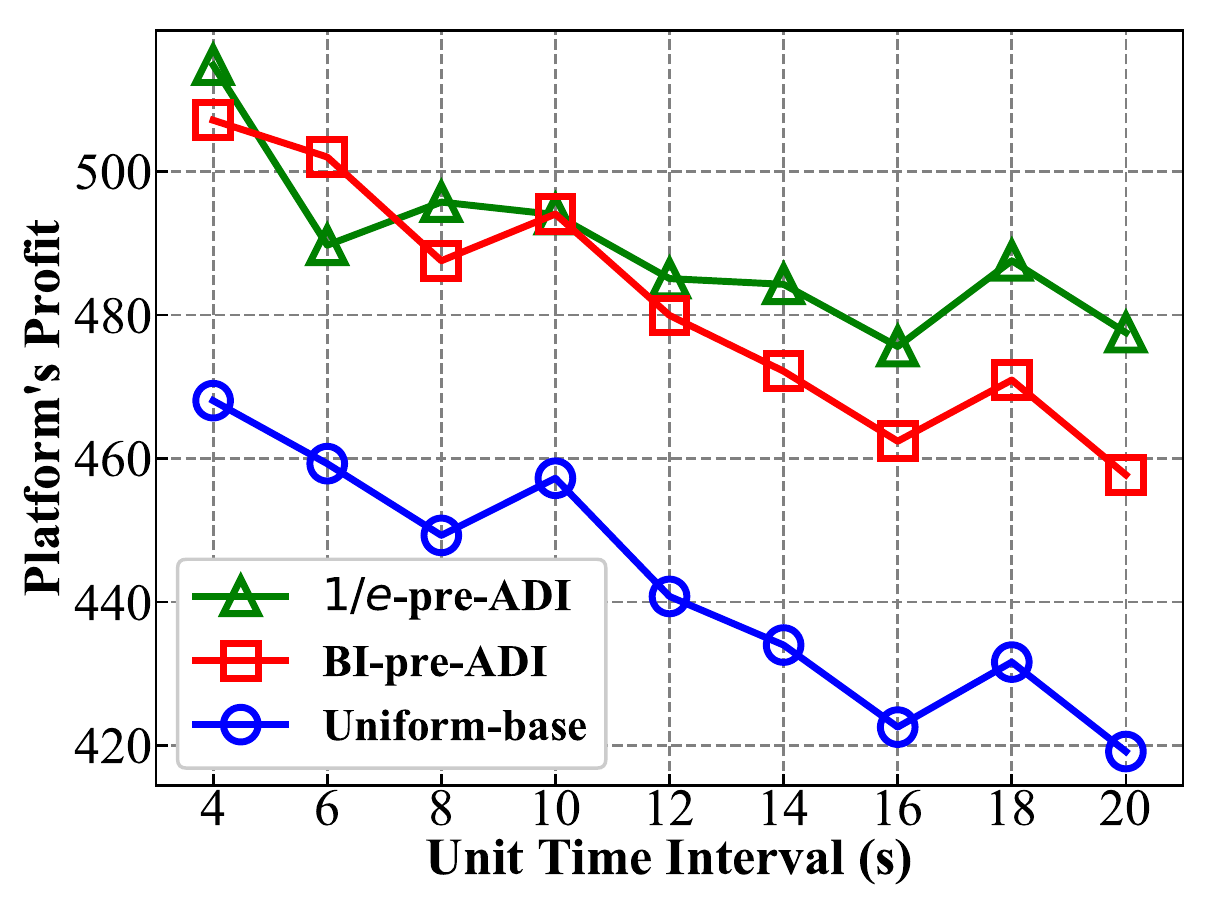}
    \end{minipage}
  }
  \subfigure[Cluster 4]{
    \begin{minipage}[t]{0.47\linewidth}
      \centering
      \includegraphics[width=\linewidth]{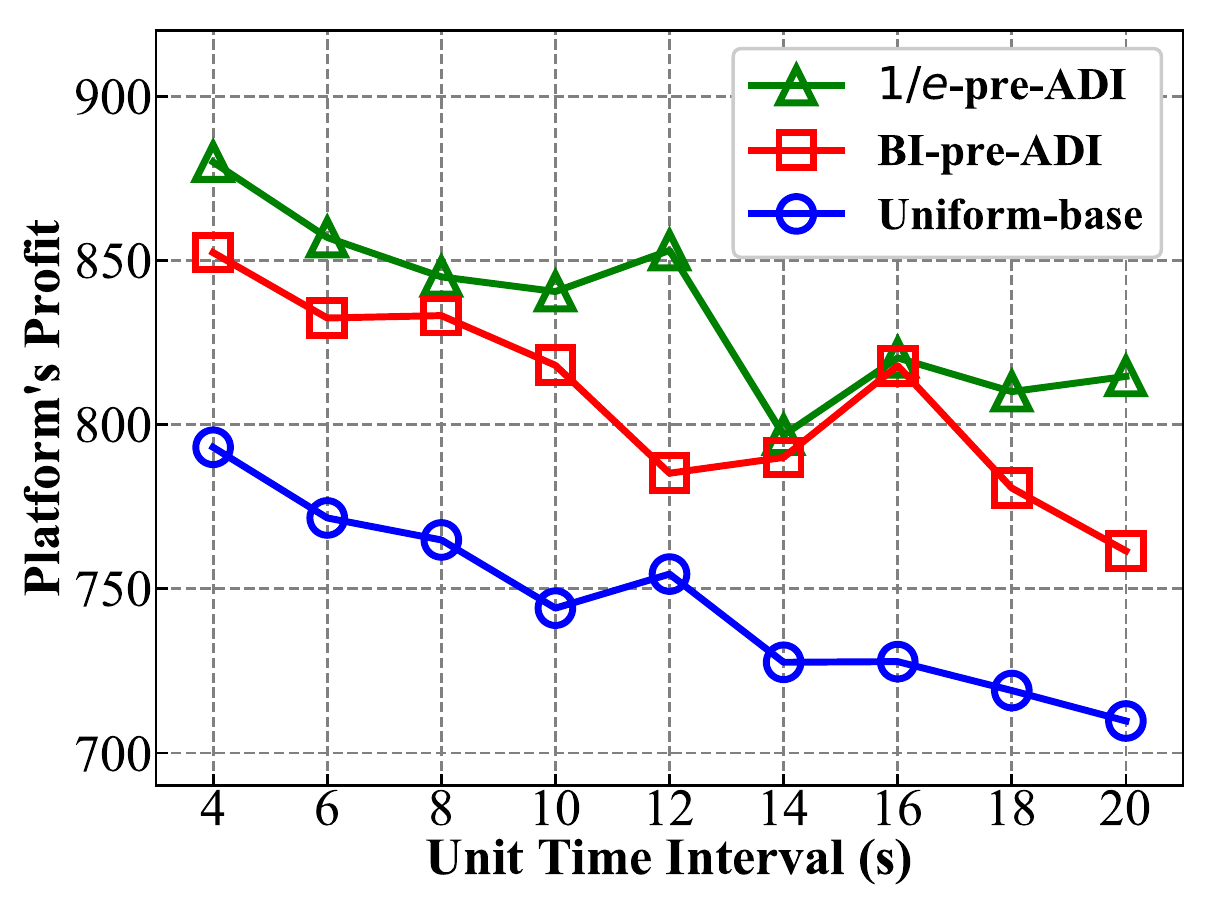}\label{int_revenue_wj}
    \end{minipage}
  }
  \centering
  \vspace{-0.4cm}
  \caption{Experiments of varying unit time interval $\Delta t$ for the pre-ADI problem.}
  \label{fig:int}
  \vspace{-0.6cm}
\end{figure}

\noindent
\textbf{Effect of varying unit time interval $\Delta t$.} Figure \ref{fig:int} shows the experiment results of varying the unit time interval $\Delta t$. The maximum dispatching interval $\beta\Delta t$ is set to be 90 seconds in the experiments. As the unit time interval $\Delta t$ increases, after deciding not to dispatch, the platform will have to wait longer to decide again, and the maximum number of time instances for the platform to decide within a dispatching interval will decrease. Figures \ref{int_revenue_hlg}-\ref{int_revenue_wj} show that the total profits obtained by $1/e$-pre-ADI and BI-pre-ADI are both higher than the profits obtained by uniform-base as the unit time interval $\Delta t$ varies.

\begin{figure}[t]
  \centering
  \subfigure[Cluster 1]{
    \begin{minipage}[t]{0.47\linewidth}
      \centering
      \includegraphics[width=\linewidth]{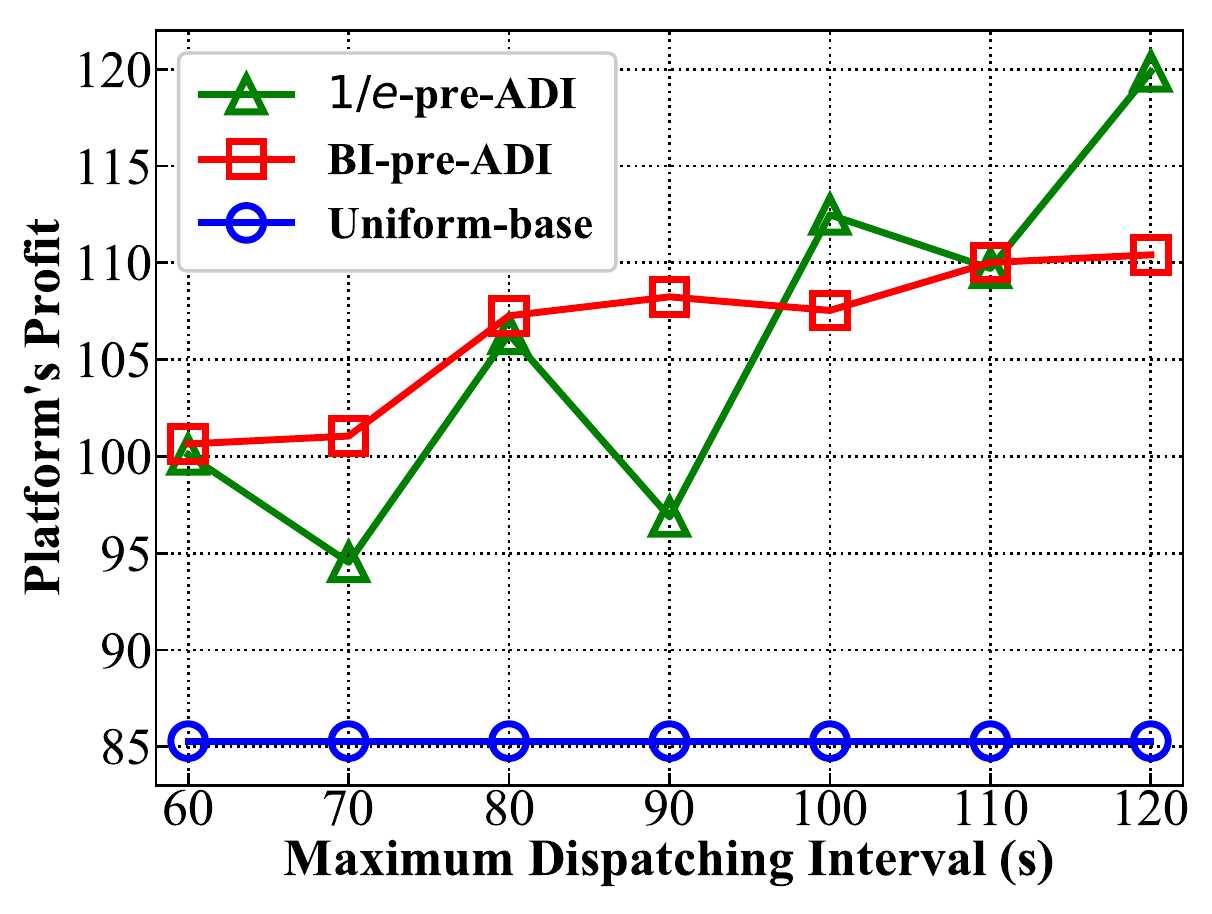}\label{max_revenue_hlg}
    \end{minipage}
  }
  \subfigure[Cluster 2]{
    \begin{minipage}[t]{0.47\linewidth}
      \centering
      \includegraphics[width=\linewidth]{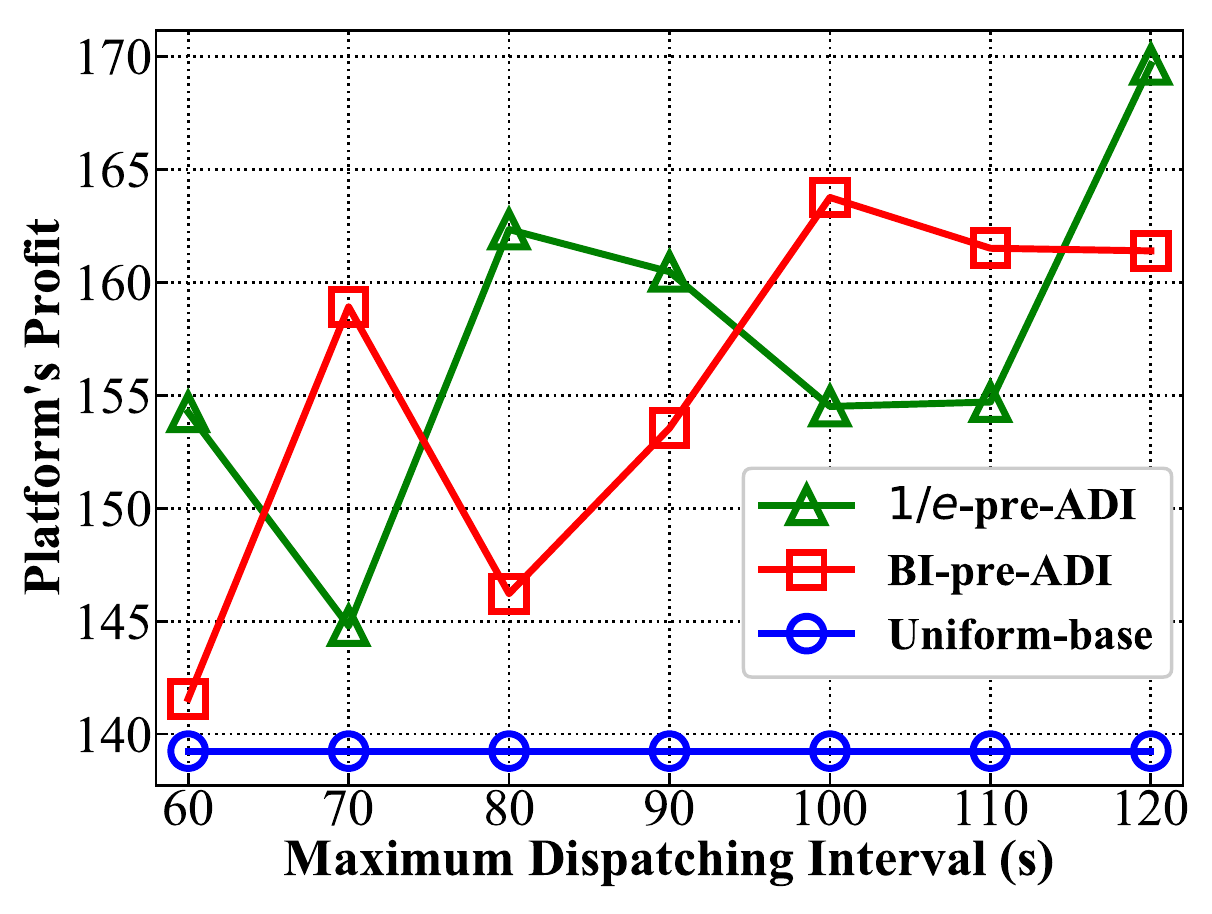}\label{max_revenue_wjxy}
    \end{minipage}
  }

  \subfigure[Cluster 3]{
    \begin{minipage}[t]{0.47\linewidth}
      \centering
      \includegraphics[width=\linewidth]{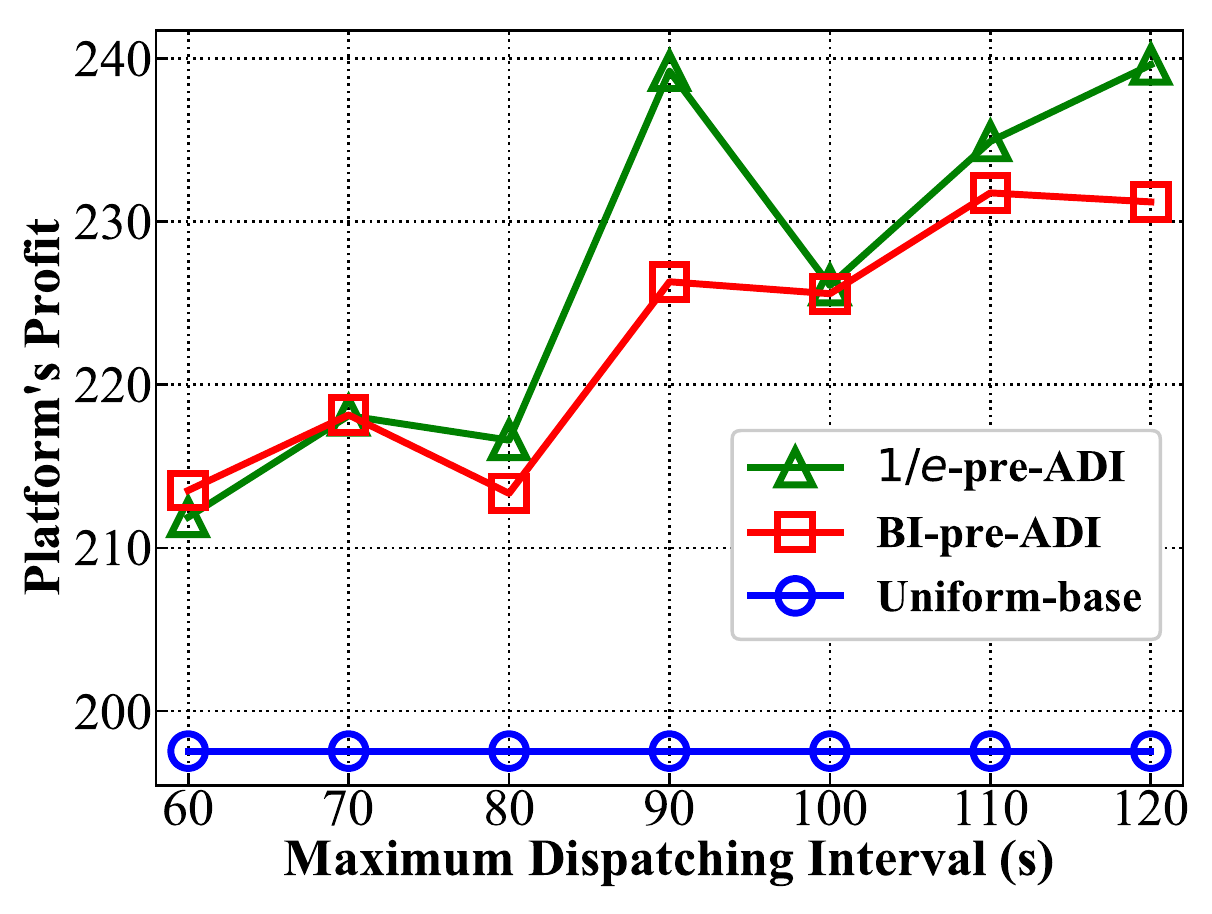}\label{max_revenue_rjy}
    \end{minipage}
  }
  \subfigure[Cluster 4]{
    \begin{minipage}[t]{0.47\linewidth}
      \centering
      \includegraphics[width=\linewidth]{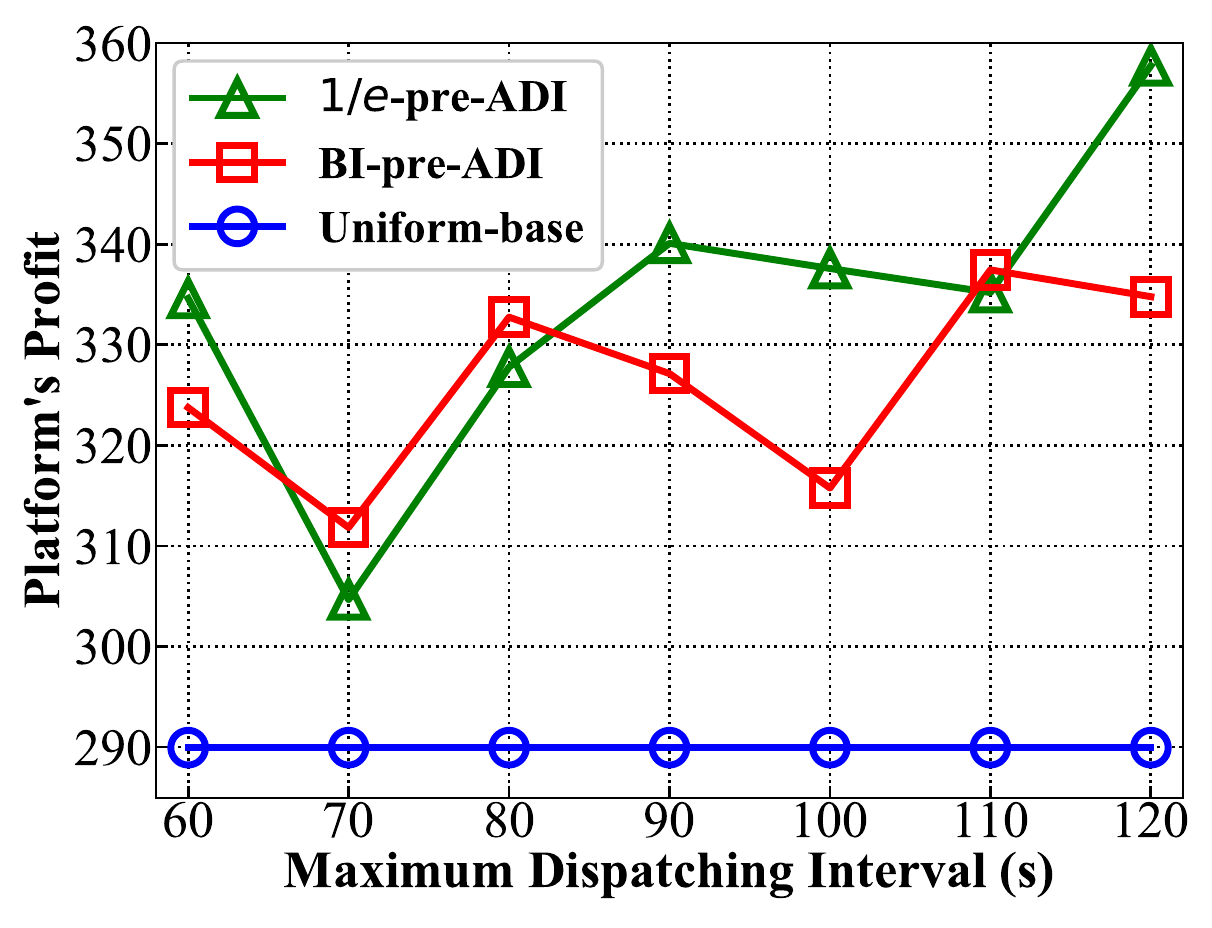}\label{max_revenue_wj}
    \end{minipage}
  }
  \centering
  \vspace{-0.5cm}
  \caption{Experiments of varying maximum dispatching interval $\beta \Delta t$ for the pre-ADI problem.}
  \label{fig:max}
  \vspace{-0.6cm}
\end{figure}

\begin{figure}[t]
 \centering
 \subfigure[Cluster 1]{
  \begin{minipage}[t]{0.47\linewidth}
   \centering
   \includegraphics[width=\linewidth]{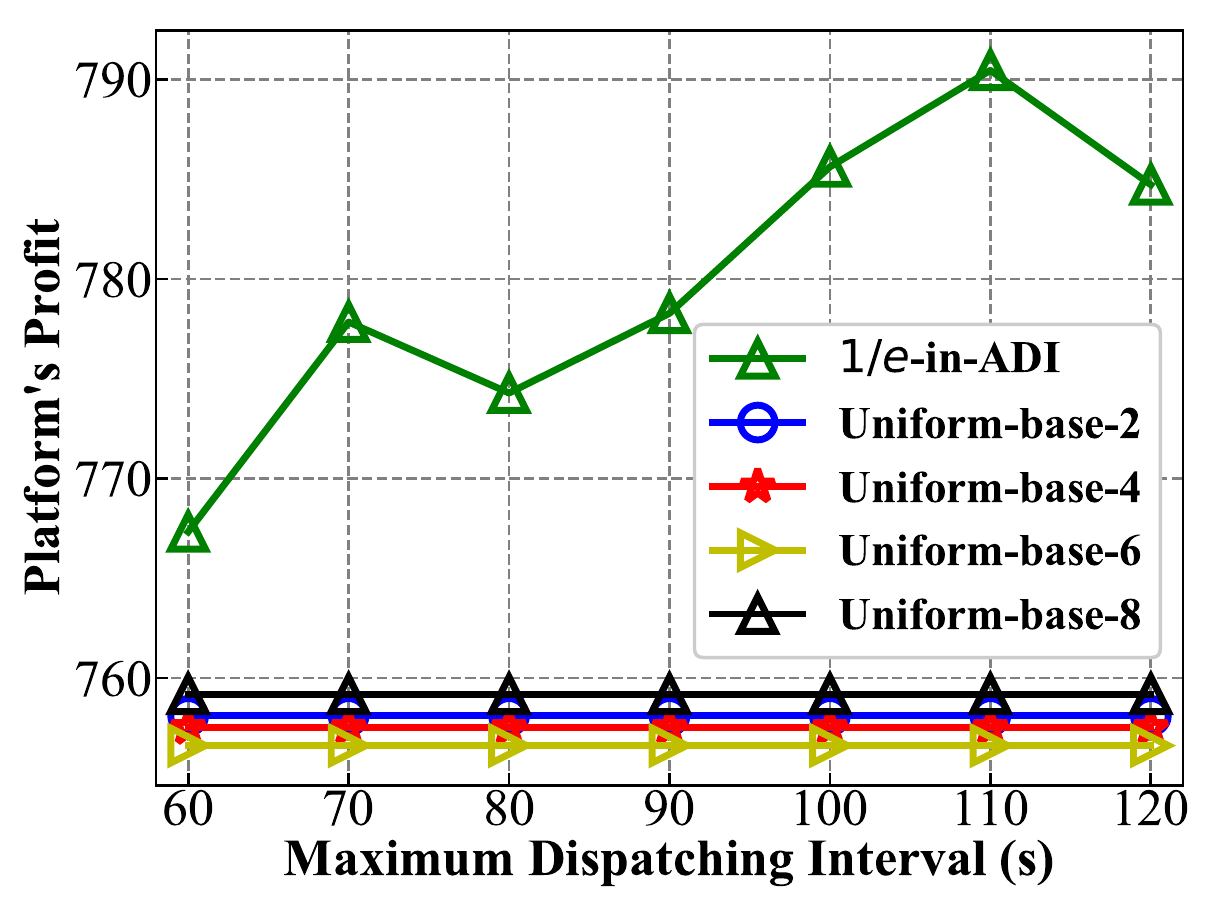}\label{in_max_revenue_hlg}
  \end{minipage}
 }
 \subfigure[Cluster 2]{
  \begin{minipage}[t]{0.47\linewidth}
   \centering
   \includegraphics[width=\linewidth]{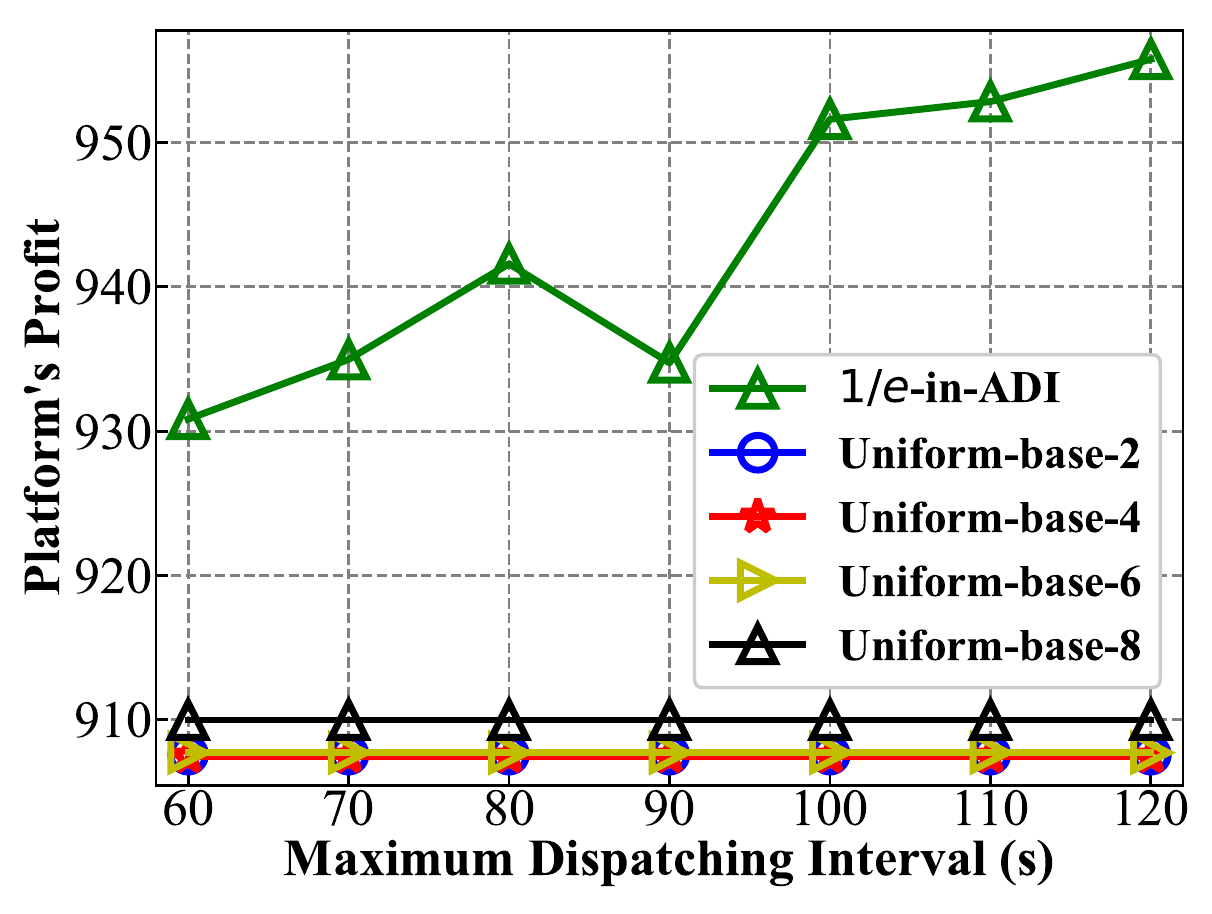}
  \end{minipage}
 }

 \subfigure[Cluster 3]{
  \begin{minipage}[t]{0.47\linewidth}
   \centering
   \includegraphics[width=\linewidth]{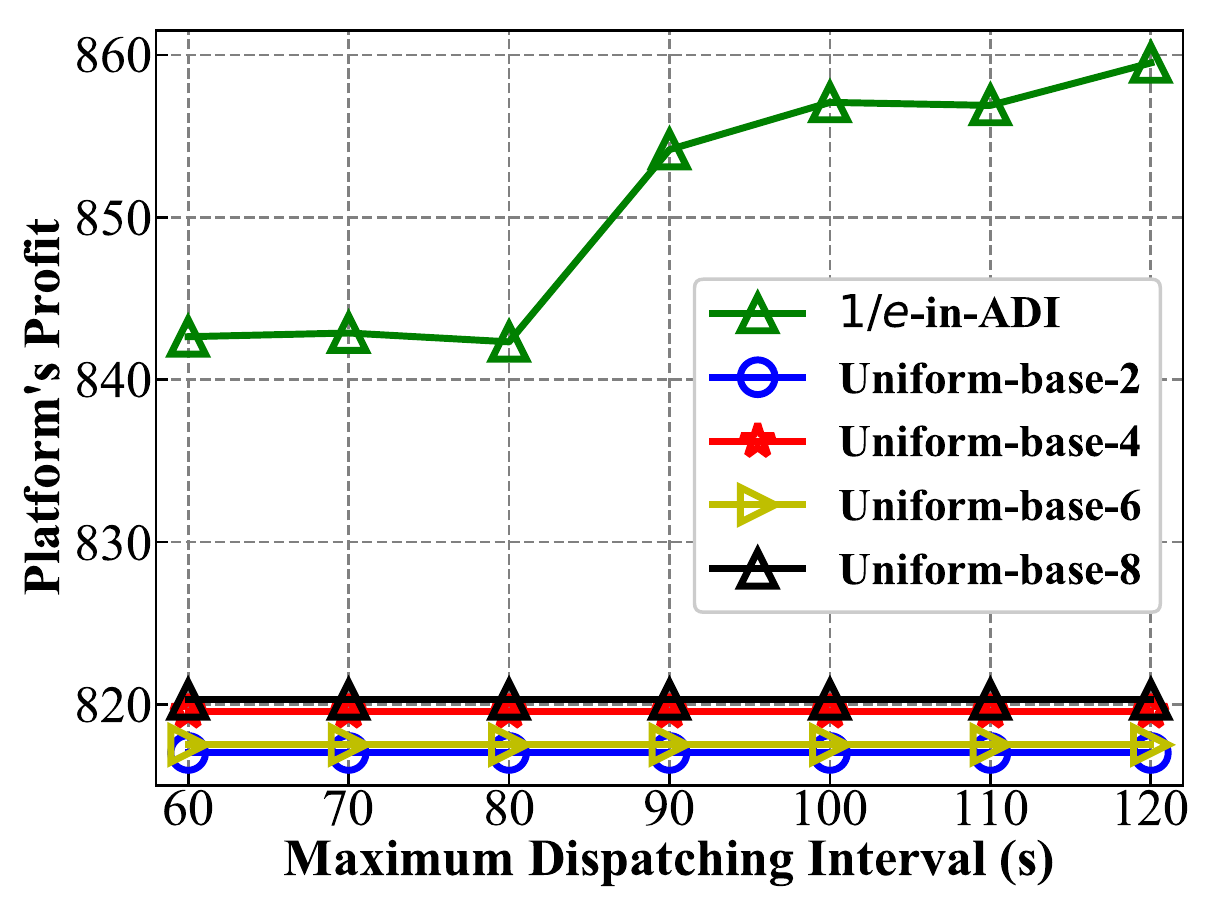}
  \end{minipage}
 }
 \subfigure[Cluster 4]{
  \begin{minipage}[t]{0.47\linewidth}
   \centering
   \includegraphics[width=\linewidth]{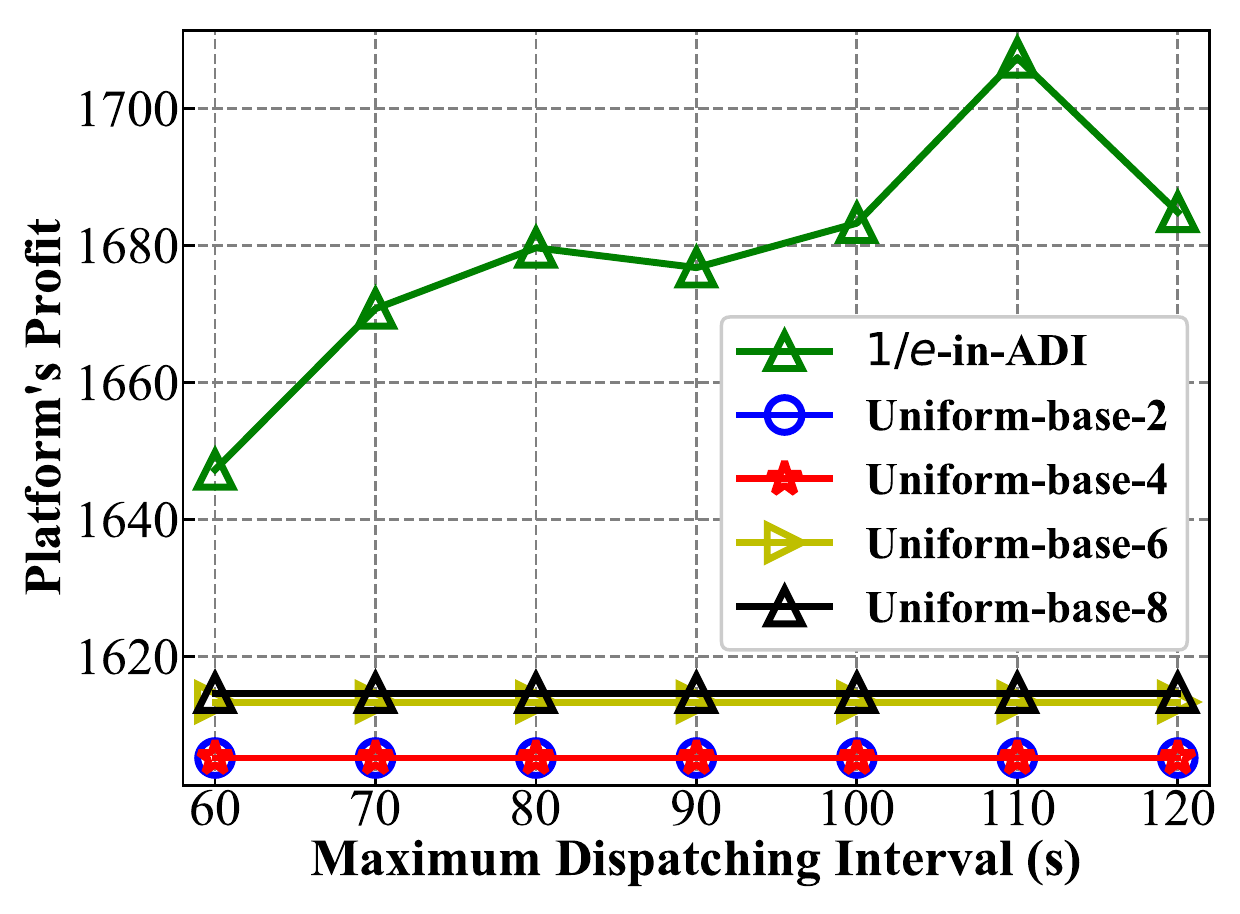}\label{in_max_revenue_wj}
  \end{minipage}
 }
 \centering
 \vspace{-0.4cm}
 \caption{Experiments of varying maximum dispatching interval $\beta \Delta t$ for the in-ADI problem.}
 \label{fig:inmax}
 \vspace{-0.6cm}
\end{figure}

\textbf{Effect of varying the maximum dispatching interval $\beta \Delta t$ for the pre-ADI problem.} The experiment results of varying the maximum dispatching interval $\beta \Delta t$ are shown in Figure \ref{fig:max}. In the experiments of varying the maximum dispatching interval, the unit time interval $\Delta t$ is set to be 20 seconds. As $\beta\Delta t$ increases, there are more unit time intervals in a dispatching interval, which allows the platform to postpone dispatching longer and have more chances to decide whether to dispatch. As shown in Figures \ref{max_revenue_hlg}-\ref{max_revenue_wj}, all platform's profits obtained by our proposed algorithms are better than profits obtained by the uniform-base method. 

Furthermore, it can be observed in Figures \ref{max_revenue_hlg}-\ref{max_revenue_wj}, the platform's profit has a trend to increase as the length of maximum dispatching interval increases. The reason is that the ridesharing system is allowed to observe more unit time intervals within a dispatching interval, and thus the probability of deciding to dispatch at the optimal time instance is greater.

\textbf{Effect of varying the maximum dispatching interval $\beta \Delta t$ for the in-ADI problem.}
Figure \ref{fig:inmax} shows the experiment results of varying the maximum dispatching interval $\beta\Delta t$ for the in-ADI problem. The unit time interval is set to be 2 seconds in the experiments. The uniform baseline methods with different unit time intervals, $i.e.$, 2, 4, 6, 8 seconds, are denoted by Uniform-base-2, 4, 6, 8, respectively. Similar to experiments for the pre-ADI problem, as the maximum dispatching interval $\beta\Delta t$ increases, there are more unit time intervals in a dispatching interval. Figures \ref{in_max_revenue_hlg}-\ref{in_max_revenue_wj} show that the profits obtained by $1/e$-in-ADI algorithm are better than profits obtained by all the uniform baselines.

Figures \ref{in_max_revenue_hlg}-\ref{in_max_revenue_wj} show that the platform's profit obtained by $1/e$-in-ADI algorithm have a trend to increase as the maximum dispatching interval $\beta\Delta t$ increases. Similar to the experiments for the pre-ADI problem, the reason is because when the ridesharing system is allowed to observe more unit time intervals within a dispatching interval, the ridesharing system is more likely to receive more shareable and profitable order pairs, and thus the ridesharing system has a greater probability to obtain more profit.

\section{Related Work}\label{Related}
With the availability of huge amount of data from transportation systems, data-driven approaches \cite{He2019BikeLanePlanning,Zhang2019BikeRecyclePlanning,He2020BikeFlow,Zhang2019UnveilingTD,Zhou2020OTDPE} have become more and more popular for urban computing tasks. Several works \cite{Xu2018KDD,Jin2019CoRide,Li2019MeanField,He2019CapsuleRL} study order dispatch and fleet management problems in ride-hailing platforms using datasets consisting of order information and vehicle trajectories. Different from previous works that dispatch one order to a driver, we focus on the ridesharing scenario where multiple orders share a ride in this paper, and investigate adapting dispatching intervals to the uneven spatio-temporal order distributions. 

Ridesharing has been a popular research topic in recent years. In this section, we organize the related works on this topic into three categories, \textit{i.e.}, dynamic pricing, route planning, and online matching.

\textbf{Dynamic Pricing.}
One line of works \cite{Asghari2018ADAPT,Fang2017Subsidies,Luca2017Nash,Tong2018DynamicPricing,Asghari2017OnlineTruthful,Asghari2016AuctionBased} investigate pricing for ridesharing systems. More specifically, Tong \textit{et al.} \cite{Tong2018DynamicPricing} design dynamic pricing strategies based on the spatio-temporal distributions of the demand and supply. Asghari \textit{et al.} \cite{Asghari2018ADAPT} consider the demand and supply at both orders’ origins and destinations for pricing. All these works aim to increase the platform’s profit, which is the same with this paper. However, instead of adapting prices, we adapt dispatching intervals to the uneven spatio-temporal order distributions, which is proven by this paper to be quite effective for boosting the platform's profit. Other works \cite{Fang2017Subsidies,Asghari2017OnlineTruthful,Luca2017Nash,Asghari2016AuctionBased} adopt an auction-based framework with different objectives, including maximizing social welfare and improving fairness compared with this paper.

\textbf{Route Planning.}
Recently, studies on route planning \cite{Ma2013Tshare,MaTshare2,Huang2014Guarantee,Khan2017Agree,Tong2018Unified,Bei2018AAAI,Ta2018SharedRoute,Jindal2018TaxiCarpool,Hargrave2017RoadCondition} for ridesharing have drawn significant attention. Among them, \cite{Ma2013Tshare,MaTshare2,Huang2014Guarantee,Khan2017Agree,Ta2018SharedRoute} focus on minimizing the travel distance. Furthermore, Tong \textit{et al.} \cite{Tong2018Unified} propose a unified formulation of route planning for shared mobility. Hargrave \textit{et al.} \cite{Hargrave2017RoadCondition} integrate dynamic road conditions into the route planning phase. Jindal \textit{et al.} \cite{Jindal2018TaxiCarpool} adopt reinforcement learning methods and choose the effective distance covered by the driver as the reward so as to minimize the travel cost. Different from these works, we focus on improving the platform's profit and exploit the advantages of having adaptive dispatching intervals, which is orthogonal to the route planning problem studied in the above works.

\textbf{Online Matching.}
Another set of related works \cite{Song2017Trichromatic,Cheng2017Utility,Zheng2018PriceAware,Cao2020Sharek,Chen2018PTrider,Chen2018PTAdynamic,Lin2019DemandAware,Miao2019RobustDispatch,Wang2019DynamicBipartite} study online matching between orders and vehicles in the ridesharing systems. Cheng \textit{et al.} \cite{Cheng2017Utility} and Miao \textit{et al.} \cite{Miao2019RobustDispatch} both take the future demand distribution into consideration in online matching without considering order cancellations. Song \textit{et al.} \cite{Song2017Trichromatic} and Wang \textit{et al.} \cite{Wang2019DynamicBipartite} treat the online matching as a graph matching problem, but ignore the uneven spatio-temporal order distributions. Zheng \textit{et al.} \cite{Zheng2018PriceAware} take order price into consideration in online matching to improve the platform's profit. A series of works \cite{Cheng2017Utility,Cao2020Sharek,Chen2018PTrider,Chen2018PTAdynamic} emphasize more on the user satisfaction and the travellers are able to choose the vehicle matching their preferences the best. Furthermore, some works \cite{Wang2016ActivityBased,Monteiro2018Alternative} propose activity-based ridesharing, which matches requests that have similar activities and accepts alternative destinations. All of these works dispatch orders at a uniform time interval, while we then exploit the advantages of the adaptive dispatching intervals.

To summarize, to the best of our knowledge, this paper is the first work that leverages the power of adapting dispatching intervals to the uneven spatio-temporal order distributions so as to boost the platform’s profit with an upper bound on the passenger waiting time.

\section{Conclusion}\label{Conclusion}
In this paper, we propose a hierarchical approach to boost the platform's profit and meanwhile guarantee the waiting time for passengers. Such hierarchical approach consists of the spatial clustering algorithm and adaptive interval algorithms. Our spatial clustering algorithm finds the spatial clusters such that the cells within each cluster are suitable to adopt the same adaptive dispatching interval. For different ridesharing models, we propose adaptive dispatching interval algorithms, i.e., the $1/e$-pre-ADI and BI-pre-ADI for the scenario where only pre-trip ridesharing is considered, and $1/e$-in-ADI algorithm in the case where in-trip ridesharing is also considered. Our adaptive interval algorithms determine dispatching time instances within given spatial clusters that improve platform's profit in an online manner. Furthermore, we prove it impossible to design constant-competitive-ratio algorithms for the online adaptive interval problem. We validate that our proposed algorithms significantly increase platform's profit compared to existing approaches with a large-scale ridesharing order dataset, which contains all of the over 3.5 million ridesharing orders in Beijing, China, received by Didi Chuxing from October 1st to October 31st, 2018. 

\bibliographystyle{IEEEtran}
\bibliography{reference}

\begin{appendices}
\section{Proof of Theorem 3.1}\label{appendix:complexityspatial}
\begin{proof}
  Suppose that $\mathcal{G}$ has $m$ connected components. For each connected component $\mathcal{S}_i=(\mathcal{V}_i, \mathcal{E}_i, w)$, it takes $O\big(|\mathcal{E}_i| \log |\mathcal{V}_i|\big)$ time to find its maximum spanning tree $T_i$ by the Prim's algorithm (line \ref{algo:sp3} in Algorithm \ref{algo:spatial}). Then, Algorithm \ref{algo:spatial} calls Algorithm \ref{algo:ED} to obtain the spatial cluster set induced from $T_i$ (line \ref{algo:sp4} in Algorithm \ref{algo:spatial}). The running time of Algorithm \ref{algo:ED} depends on both whether the variance of the edge weights of $T_i$ is below $\theta$, and whether the partitioning of $T_i$ is balanced. If the variance of the edge weights is below $\theta$, it costs $O(1)$ time. Otherwise, the worst case occurs when the deletion routine produces one tree with $|\mathcal{V}_i|-1$ vertices and the other with 1 vertex. Let us assume that this unbalanced partitioning arises in each recursive call, then the running time is $O\big(|\mathcal{V}_i|^2\big)$. Therefore, the total running time in the worst case is
  
  \begin{align*}
    &\sum_{i=1}^m O\big(|\mathcal{V}_i|^2\big) + O\big(|\mathcal{E}_i| \log |\mathcal{V}_i|\big)\\
    =&\sum_{i=1}^m O\big(|\mathcal{V}_i|^2 \log |\mathcal{V}_i|\big)
    = \sum_{i=1}^m O\big(|\mathcal{V}_i|^3\big) = O\Bigg(\bigg(\sum_{i=1}^m |\mathcal{V}_i|\bigg)^3\Bigg)\\
    =&O\big(|\mathcal{V}|^3\big),
  \end{align*}
  which proves the result stated in Theorem \ref{complexityspatial}.
\end{proof}

\section{Proof of Lemma 4.1}\label{appendix:lem1}
\begin{proof}
    Suppose there exists an algorithm $A$ that is $\alpha$-competitive ($\alpha\leq 1$). Denote the profit function of the pre-ADI problem by $p(a, x)$, where $a$ is a deterministic algorithm to solve this problem and $x$ is an appropriate input. For the pre-ADI problem, define the profit function to be platform's profit. Then $p(a,x)$ measures the performance of the algorithm $a$ on the pre-ADI problem. Since the pre-ADI problem is a maximization problem, it implies that for any input $x$, platform's profit obtained by the algorithm $A$, $p(A, x)$, is at least $\alpha$ times of that obtained by the optimal offline algorithm $OPT$ on the same input, $p(OPT, x)$. Therefore, it only takes one input that renders the algorithm $A$ not $\alpha$-competitive to counter the claim.
    
    Assume there exists an adversary model, which is aware of all the decisions that algorithm $A$ makes and generates input for the algorithm $A$. The adversary generates such an input $x$ that only one order $r$ is raised during the entire time period. The order $r$ is raised in one unit time interval and canceled in the next. The order $r$ has such origin, destination, and raised time, that algorithm $A$ would decide not to dispatch immediately at the end of the unit time interval in which the order $r$ is raised. It is obvious that the optimal offline algorithm would dispatch at the end of unit time interval where order $r$ is raised, and obtain the profit, $\textnormal{Profit}_r$, from $r$. The profit of the optimal offline algorithm is, $p(OPT, x) = \textnormal{Profit}_r$. However, since algorithm $A$ would decide not to dispatch, it would complete no order. The profit of algorithm $A$ is zero, $i.e.$, $p(A, x) = 0$. Hence, the competitive ratio of algorithm $A$ on this input is:

    \begin{equation*}
        \frac{p(A, x)}{p(OPT, x)} =  \frac{0}{\textnormal{Profit}_{r}} = 0.
    \end{equation*}
    The competitive ratio of algorithm $A$ is 0, which is contrary to the claim that it is $\alpha$-competitive ($\alpha\leq 1$).
\end{proof}

\section{Proof of Theorem 4.2}\label{appendix:thr1}


\begin{proof}
    Having proved lemma \ref{lem1}, now we only need to show that there exists no randomized algorithm for the pre-ADI problem that has a constant competitive ratio $\alpha$, $\alpha\leq 1$.
    
    Firstly, we need to introduce a distribution of input, and show that the expectation of competitive ratio of any deterministic algorithm on the input distribution is not constant. Then, by applying Yao's principle, we can conclude that no randomized algorithm on this input distribution has constant competitive ratio. 
    
    
    Define a probability distribution $\mathcal{X}$ of input as follows: \textit{(i)} Assume that an order $r_0$ is raised in the first unit time interval $[0,\Delta t)$. \textit{(ii)} Assume that there are $n$ orders $r_1,\dots, r_n$, and for each $i$,  $1\leq i\leq n$, the order $r_i$ is raised in the $i$th unit time interval, $i.e.$, $[(i-1)\Delta t, i\Delta t)$. All the $n$ orders, $r_1,\dots, r_n$, are not canceled before $n\Delta t$. \textit{(iii)} Assume that any two orders in $\{r_1,\dots, r_n\}$ cannot be dispatched to the same vehicle, and any order in $\{r_1,\dots, r_n\}$ can be dispatched to the same vehicle with $r_0$. \textit{(iv)} Assume that if $i>j$, $1\leq i,j \leq n$, then $\textnormal{Profit}_{r_0,r_i} > \textnormal{Profit}_{r_0,r_j}$. Neglect the profit $\textnormal{Profit}_{r_i}$, $1\leq i\leq n$. \textit{(v)} Assume that there are $n$ different inputs $X_1,\dots, X_{n}$ in $\mathcal{X}$ with identical probability $\frac{1}{n}$. For an input $X_i$, the order $r_0$ is canceled in the $(i+1)$th unit time interval.
    
    Let $\mathcal{A}$ be the set of all deterministic algorithms that can solve the pre-ADI problem, and $A$ be a random variable of $\mathcal{A}$. For maximization problems, Yao's principle \cite{DBLP:conf/focs/Yao77} states that: 
    \begin{equation*}
        \min_{x\in \mathcal{X}}\mathbb{E}[p(A, x)] \leq \max _{a\in \mathcal{A}}\mathbb{E}[p(a, X)],
    \end{equation*}
    where $X$ is a random input chosen from $\mathcal{X}$. Hence, to show the expected competitive ratio of the randomized algorithm for the worst case could be not constant, we only need to show that the competitive ratio of the best deterministic algorithm on input distribution $\mathcal{X}$ is not constant. Define the profit function to be platform's profit of an algorithm.
    
    The optimal offline algorithm $OPT$ will dispatch at $i\Delta t$ on the input $X_i$. Therefore, 
    \begin{equation*}
        \mathbb{E}_{\mathcal{X}}[p(OPT, X)] = \frac{1}{n}\times \sum_{i = 1}^n \textnormal{Profit}_{(r_0, r_i)}.
    \end{equation*}
    
    Note the fact that the deterministic algorithm will make decisions based only on the raised orders. If the order $r_0$ is not canceled at time $i\Delta t$, $1\leq i\leq n$, then no deterministic algorithm will make different decisions on whether to dispatch on different inputs. Therefore, no deterministic algorithm can obtain the maximal total net profit on all inputs as the optimal offline algorithm. The best deterministic algorithm $ALG$ on this input distribution would dispatch at time $n\Delta t$ on all inputs with the expected profit
    \begin{equation*}
        \mathbb{E}_{\mathcal{X}}[p(ALG, X)] = \frac{1}{n}\textnormal{Profit}_{(r_0,r_n)}.
    \end{equation*}
    Hence, 
    \begin{equation*}
        \max_{a\in \mathcal{A}}\mathbb{E}_{\mathcal{X}}[p(a, X)] = \mathbb{E}_{\mathcal{X}}[p(ALG, X)] = \frac{1}{n}\textnormal{Profit}_{(r_0,r_n)}.
    \end{equation*}
    The competitive ratio of the best deterministic algorithm on the input distribution is 
    \begin{equation*}
        \frac{\mathbb{E}_{\mathcal{X}}[p(ALG, X)]}{\mathbb{E}_{\mathcal{X}}[p(OPT, X)]} = \frac{\frac{1}{n}\textnormal{Profit}_{(r_0,r_n)}}{\frac{1}{n}\times \sum_{i = 1}^n \textnormal{Profit}_{(r_0, r_i)}} = \frac{\textnormal{Profit}_{(r_0,r_n)}}{\sum_{i = 1}^n \textnormal{Profit}_{(r_0, r_i)}}.
    \end{equation*}
    Let $n$ and $\{\textnormal{Profit}_{(r_0,r_1)},\dots,\textnormal{Profit}_{(r_0,r_n)}\}$ be such that $\textnormal{Profit}_{r_0,r_n} < \alpha\times \sum_{i}^n\textnormal{Profit}_{(r_0,r_i)}$. Then, 
    \begin{equation*}
        \frac{\mathbb{E}_{\mathcal{X}}[p(ALG, X)]}{\mathbb{E}_{\mathcal{X}}[p(OPT, X)]} < \alpha.
    \end{equation*}
    Hence, by Yao's principle, there is no randomized algorithm that has constant competitive ratio for the pre-ADI problem.
\end{proof}

\section{Proof of Theorem 4.3}\label{appendix:thr2}
\begin{proof}
    First of all, we can show that there exists no deterministic algorithm for the in-ADI problem that has a constant competitive ratio. Then, we give the complete proof of theorem \ref{thr2}. 
    
    Suppose there exists an algorithm $A$ that has constant competitive ratio $\alpha$, $0<\alpha\leq 1$, for the \textit{in-ADI problem}. Consider the adversary which generates the same input as in the proof of lemma \ref{lem1}, and it is obvious that the competitive ratio of the algorithm $A$ on this input is zero as well. Hence, there exists no deterministic algorithm that has a constant competitive ratio $\alpha$.
    
    To prove that there exists no randomized algorithm that is $\alpha$ competitive, $0<\alpha\leq 1$, we need to define an input distribution on which the best deterministic will fail to have the constant competitive ratio $\alpha$, and apply Yao's principle again. 

    Define the same probability distribution $\mathcal{X}$ of input as the probability distribution of input defined in the proof of theorem \ref{thr1}. If the platform dispatches at time $i\Delta t$, $1\leq i\leq n$, then the number of passengers reaches the capacity of vehicles. Therefore, there is no in-trip ridesharing. Thus, the profit of the best deterministic algorithm $ALG$ on this input distribution $\mathcal{X}$ is the profit obtained by dispatch at time $n\Delta t$, $i.e.,$
    \begin{equation*}
      \max_{a\in \mathcal{A}}\mathbb{E}_{\mathcal{X}}[p(a,X)] = \mathbb{E}_{\mathcal{X}}[p(ALG, X)] = \frac{1}{n}\textnormal{Profit}_{(r_0, r_n)}.
    \end{equation*}

    The profit of the optimal offline algorithm $OPT$ is given by:
    \begin{equation*}
      \mathbb{E}_{\mathcal{X}}[p(OPT, X)] = \frac{1}{n}\sum_{i=1}^n\textnormal{Profit}_{(r_0,r_i)}.
    \end{equation*}

    Hence, the competitive ratio of the best deterministic algorithm on this input distribution is:
    \begin{equation*}
      \frac{\mathbb{E}_{\mathcal{X}}[p(ALG, X)]}{\mathbb{E}_{\mathcal{X}}[p(OPT,X)]} = \frac{\textnormal{Profit}_{(r_0, r_n)}}{\sum_{i=1}^n\textnormal{Profit}_{(r_0,r_i)}}.
    \end{equation*}

    By selecting $n$ and $\{\textnormal{Profit}_{(r_0,r_1)},\dots, \textnormal{Profit}_{(r_0,r_n)}\}$ such that 
    \begin{equation*}
      \textnormal{Profit}_{(r_0,r_n)}<\alpha\times \sum_{i=1}^n{\textnormal{Profit}_{(r_0,r_i)}},
    \end{equation*}

    the best deterministic algorithm $ALG$ is not $\alpha$ competitive on this input distribution. By Yao's principle, there is no randomized algorithm that has constant competitive ratio $\alpha$. 
\end{proof}
\end{appendices}
\end{document}